\newcommand{\onlineversion}[2]{#1}{} 

\onlineversion{
\documentclass[11pt,letterpaper]{article}

\topmargin=-0.35in \topskip=0pt \headsep=15pt \oddsidemargin=0pt
\textheight=9in \textwidth=6.52in \voffset=0in
}{
\documentclass[prodmode,acmec]{ec12-acmsmall}
\acmYear{2012}
\acmMonth{06}

}

\usepackage{amsmath,amsfonts,graphpap,amscd,mathrsfs,graphicx,lscape}
\usepackage{epsfig,amssymb,amstext,xspace}
\onlineversion{
\usepackage{theorem}
}{}
\usepackage{algorithm,algorithmic}
\usepackage{paralist}
\usepackage{color}              
\usepackage{epsfig}
\usepackage{bbm}

\newcommand{\E}{\mathbb{E}}
\newcommand{\abs}[1]{\vert{#1}\vert}

\onlineversion{
\newtheorem{theorem}{Theorem}[section]
\newtheorem{lemma}[theorem]{Lemma}
\newtheorem{example}[theorem]{Example}

\newtheorem{remark}[theorem]{Remark}
}{}

\newtheorem{claim}[theorem]{Claim}

\newtheorem{defn}[theorem]{Definition}
\newtheorem{observation}[theorem]{Observation}

\newtheorem{open_problem}[theorem]{Open Problem}

\newtheorem{fact}[theorem]{Fact}

\def\squareforqed{\hbox{\rule{2.5mm}{2.5mm}}}

\def\QED{\ifmmode\squareforqed 
  \else{\nobreak\hfil   
    \penalty50                 
    \hskip1em                  
    \null                      
    \nobreak                   
    \hfil                      
    \squareforqed              
    \parfillskip=0pt           
    \finalhyphendemerits=0     
    \endgraf}                  
  \fi}

\def\blksquare{\rule{2mm}{2mm}}
\def\qedsymbol{\blksquare}
\newcommand{\bg}[1]{\medskip\noindent{\bf #1}}
\newcommand{\ed}{{\hfill\qedsymbol}\medskip}

\onlineversion{
\newenvironment{proof}{\bg{Proof : }}{\ed}
}{}
\newenvironment{proofof}[1]{\textbf{Proof of #1 : }}{\ed}

\newcommand{\R}{\ensuremath{\mathbb R}}

\newcommand{\poly}{\operatorname{poly}}

\newcommand{\buyersurplus}{\xi}

\newcommand{\comment}[1]{}

\newcommand{\junk}[1]{}

\newcommand{\prob}{\mathbb{P}}

\newlength{\tmp} \newlength{\lpsx} \newlength{\lpsy} \newlength{\upsx}
\newlength{\upsy}

\newcommand{\note}[1]{{\color{blue}\emph{#1}}}

\newcommand{\xhdr}[1]{\vskip 6pt \noindent {\bf #1.}}
\newcommand{\one}{\mathbbm{1}}
\newcommand{\sspace}{\mathcal{S}}

\onlineversion{
\newcommand{\email}[1]{\texttt{#1}}
}{}

\begin{document}

\onlineversion{

\setcounter{page}{0}

\title{Optimal Mechanisms for Selling Information}

\author{
Moshe Babaioff\\
       Microsoft Research SVC\\
       \email{moshe@microsoft.com}
\and
Robert Kleinberg\\
       Cornell University\\
       \email{rdk@cs.cornell.edu}
\and
Renato Paes Leme\thanks{This work was done while the author was an intern in
Microsoft Research SVC. He is supported by a Microsoft Research Fellowship.
}\\
       Cornell University\\
       \email{renatoppl@cs.cornell.edu}
}

}{


\markboth{Babaioff, Kleinberg and Paes Leme}{Optimal Mechanisms for Selling
Information}

\title{Optimal Mechanisms for Selling Information}
\author{Moshe Babaioff
\affil{Microsoft Research Silicon Valley}
Robert Kleinberg
\affil{Cornell University}
Renato Paes Leme
\affil{Cornell University}}

}

\onlineversion{
\maketitle
}{}

\begin{abstract}
The buying and selling of information is taking place at a
scale unprecedented in the history of commerce, thanks to
the formation of online marketplaces for user data.
Data providing agencies sell user information to advertisers
to allow them to match ads to viewers more effectively.
In this paper we study the design of optimal
mechanisms for a monopolistic data provider
to sell information to a buyer, in a model where
both parties have (possibly correlated) private
signals about a state of the world, and the buyer
uses information learned from the seller, along with
his own signal, to choose an action
(e.g., displaying an ad) whose payoff depends on
the state of the world.

We provide sufficient conditions under which there is a simple one-round
protocol (i.e. a protocol where the buyer and seller each sends a single message,
and there is a single money transfer) achieving optimal revenue. In
these cases we present a polynomial-time algorithm that computes the optimal mechanism.
Intriguingly, we show that multiple rounds of partial information disclosure (interleaved by payment to the seller) are sometimes necessary to achieve optimal revenue if the buyer is allowed to abort his interaction with the seller prematurely.
\comment{
We formulate a version of the revelation principle
in this setting and provide sufficient conditions
under which it holds, along with a polynomial-time
algorithm for computing the optimal mechanism in these
cases.
Intriguingly, we show
that the revelation principle may \emph{fail to hold}
in situations where the buyer is allowed to abort his
interaction with the seller prematurely.
}
We also prove some negative results about the
inability of simple mechanisms for selling information
to approximate more complicated ones in the worst case.
\end{abstract}

\onlineversion{
\renewcommand{\thepage}{}
\clearpage
\pagenumbering{arabic}
}{






\begin{bottomstuff}
This work was done while Renato Paes Leme was an intern in Microsoft Research
SVC. During the academic year he is supported by a Microsoft Research
Fellowship.
\end{bottomstuff}
}

\onlineversion{}{



\maketitle
}

\section{Introduction}
\label{sec:intro}

A growing trend in online advertisement is the usage of
behavioral targeting and user information (like demographics)
to better match advertisements to the viewer. This is possible due to the
existence of data providing agencies, like Bluekai, Bluecava, eXelate
Media, Clearsprings and RapLeaf, whose business consists in collecting,
curating and selling information about {\em user intent} to advertisers. An
article in NYT \cite{NYT_article_bluekai} analyzes this phenomenon and points
out that data agencies are not exclusively an Internet phenomenon.
For example, for many years companies
like Acxiom and Experian (founded in 1969 and 1980 respectively) have been
collecting information about consumer habits and selling this
information to marketers, who then can use it to send catalogs by mail.

A concrete situation is as follows: an advertiser
has multiple different ads that he can present to the viewer,
and the effectiveness of each of them depends on both the ad
and the characteristics of the viewer. For example, a car maker
would rather show sport car advertisement to affluent young
bachelors while showing ads for family cars to older viewers
with kids. A data providing agency (the seller), might have
some information about the viewer generating the impression,
like gender, age and past interaction on that site. Such
information could be valuable to the advertiser as he would
be able to use it for better targeting, and the monopolist
seller would like to extract as much as possible out of this
value as revenue. The advertiser (buyer) might have some
information about the viewer as well, and this information
might possibly be correlated with the seller's information.
The seller has uncertainty about the buyer's information or
utilities, yet possesses some belief about those.

While selling information about viewers raises
obvious privacy questions, it also raises
fascinating questions of a purely economic nature.
How does one quantify the value of this information?
What is the optimal (i.e.\ revenue-maximizing) selling
strategy for information?
What are the qualitative differences between
selling information and selling physical goods and services?
How do these differences influence the design of
markets for information, and the algorithmic
problems underlying such markets?

To highlight the issues inherent in such questions,
it is helpful to highlight some differences between
a seller offering $n$ distinct goods for sale
and a seller offering $n$ bits of information.
\begin{compactenum}[(1)]
\item  A seller of goods can group them into
bundles, offering a subset of the goods at a specified
price.  A seller of bits can do many other things:
for example, she\footnote{Throughout this paper,
we use female pronouns for the seller
and male pronouns for the buyer.} can set a specified price for
revealing the Boolean XOR of the first two bits or some more complex function
of the bits.
\item  A consumer of goods generally knows
their value even before they are allocated.
The value of a piece of information is typically not
known until the information is revealed.
\item  By the taxation principle, a buyer of goods
can be assumed, without loss of generality,
to be facing a posted-price for each bundle (that is independent of his type).
A seller of
information may, in some cases, be able to extract
strictly more revenue using an interactive protocol\footnote{In this paper we will use the terms {\em protocol} and {\em mechanism} interchangeably.}
rather than posted pricing.  (See Section~\ref{sec:uncommitted}.)
\end{compactenum}
To be sure, there are some cases in which trading goods
possesses some of the characteristics of trading
information noted above.  For example, a customer
in a restaurant does not necessarily know the
quality of the food he is about to consume; in turn,
this can lead to sellers using interactive protocols,
for example, allowing the restaurant
customer to try a limited sample of food
for a reduced price (or even for free) before
deciding whether to order more.  We interpret
such situations as markets in which information
and goods are coupled together, i.e.\ revealing
the quality of the food occurs in conjunction
with selling the food itself.

This paper addresses some of these questions
raised above, by situating them in a model that eliminates
extraneous features --- such as coupling of goods and
information, or competition between
multiple buyers and sellers of information --- while
attempting to remain quite general in the model's
assumptions about information and its utility.
Our only such assumption is that the
utility of information lies in guiding future
actions of the party receiving the information.
Thus, in our model there is a single seller and
single buyer.  A state of the world (denoted by $\omega$
henceforth) is known to the seller
but not the buyer\footnote{Our model also
incorporates cases in which the buyer and seller
\emph{each} receive (possibly different) signals about the
state of the world.  In such cases, we simply treat
the buyer's signal as being encoded in her payoff type,
$\theta$.}, and the buyer's payoff
type (denoted by $\theta$)
is known to the buyer but not the seller.
The two parties engage in an
interactive protocol, consisting of one or more
rounds in which signals and/or money are exchanged.
After this interaction, the buyer chooses an
action (based on his posterior beliefs)
and receives a payoff
that depends on the state of the world, his own
payoff type, and the chosen action.

Crucially, we assume that the seller designs the protocol
and can be trusted to faithfully follow the protocol he designs.
The buyer, on the other hand, need not be honest:
he may send signals
that are inconsistent with his true payoff type if it
is rational to do so.
We do, however, distinguish between \emph{committed} buyers ---
who can be committed to complete the specified protocol even
if they are sending dishonest signals ---
and \emph{uncommitted} buyers, who may abort the protocol
if it is rational to do so, for example when they have received
information and not yet paid for it.

\xhdr{Our results}
The set of all interactive protocols is a large and ill-structured
space.  Searching for the revenue-maximizing one is unfathomably complex
unless there is a way to limit the search space. Our first set of results
provide the tools necessary for that. In mechanism
design this is often done by invoking some form of the
\onlineversion{\emph{revelation principle}, due to Gibbard \cite{gibbard-73} and
Myerson \cite{myerson-79}}{
\emph{revelation principle}~\cite{gibbard-73,myerson-79}}. In their setting,
buyers have private types and the seller (mechanism designer) needs to choose among a
set of outcomes and can charge payments from the buyers. The \emph{revelation
principle} states that if a certain outcome and payments can be implemented in
equilibrium of a possibly complicated and interactive mechanism, then it can be
implemented in a simple direct revelation mechanism, where buyers report their types,
and the mechanism chooses an outcome and payments. Moreover, this mechanism has
a simple equilibrium where each buyer reports his type truthfully. If the
outcome is the allocation of traditional goods, the revelation principle
implies that the mechanism can be implemented as a protocol consisting of three
steps: (i) buyers report their type, (ii) payments take place,
(iii) outcome is determined.
We say that such a mechanism has the \emph{one-round revelation}
property, a property which we define precisely later (in
Definition \ref{def:one_round_revelation}), but intuitively it means that the
buyers move only once (by declaring their type), payments happen only once and
the sellers move only once (by choosing the outcome). 

Now, consider the case where instead of an allocation of traditional goods
the outcome is the disclosure of information. Myerson's revelation principle
still holds in the sense that any outcome can be implemented by a mechanism
where buyers truthfully report their type in the first step. It is not clear,
however, if the stronger property of \emph{one-round revelation} still holds.
After the buyer report his type, a sequence of payments and partial information
disclosures might be required in order to implement a certain outcome.

Our first set of results (Theorems~\ref{thm:rev-principle-independent}
and~\ref{thm:rev-principle-correlated}) provide conditions under which the
one-round revelation property holds. Their precise theorem statement,
to be given in Sections~\ref{sec:independent} and~\ref{sec:sec:correlated},
are a bit stronger: they supplement the revelation principle with additional
information about the relative timing of signals and payments.

\begin{theorem} \label{thm:rev-principle-summary}
When buyers are committed, or when buyers are uncommitted but
$\omega$ and $\theta$ are independent random variables,
any mechanism can be transformed into a mechanism that
extracts the same revenue and has the following form:
the buyer and seller each sends a single message, payment
takes place only once, the buyer's message is simply an
announcement of his type, and truthful reporting maximizes
the buyer's utility.
\end{theorem}

Interestingly, the revelation principle \emph{fails} in the remaining case, when there are uncommitted buyers and correlated signals.  The usual logic justifying the revelation principle --- that the agents can always report their types to the mechanism and let it simulate their optimal strategy given their type --- does not apply for a subtle reason having to do with the timing of payments, the correlation of the signals, and the fact that the buyer is uncommitted. The direct mechanism that attempts to simulate an interactive protocol cannot determine an unbiased estimate of the buyer's expected payment before observing $\omega$, because, unlike in the independent case, the conditional distribution of $\omega$ depends on the value of $\theta$ (the buyer's \emph{true} type) and not necessarily on the type that is reported. On the other hand, if the mechanism simulates the protocol using the realization of $\omega$ and posts a price that depends on the simulation outcome,  this fails because the buyers are uncommitted: the price  reveals information about $\omega$, and the buyers may take this information for free while refusing to pay.

Our next  results concern algorithms for
computing the optimal mechanism.  Even when the
one-round revelation property holds, it is far from obvious
how to compute the optimal mechanism efficiently.
In a one-round revelation mechanism, the seller
allocates information to the buyer by revealing
a (possibly random) signal sampled from a distribution
that depends on $\omega$ and $\theta$.
The main difficulty is that seller is free to choose
the support size of this distribution (i.e., the number of
potential signals) and in principle, this leads to an
optimization problem of unbounded dimensionality.
Nevertheless, we show that the optimal mechanism can
be computed in polynomial time by solving a convex
program of bounded dimensionality; a by-product of the
proof is an explicit upper bound on the number of potential
signals.
\begin{theorem} \label{thm:algorithms-summary}
Suppose that $\omega$ can take only $m$ possible values and
$\theta$ can take only $n$ possible values.
When buyers are committed, or when buyers are uncommitted but
$\omega$ and $\theta$ are independent random variables,
there is an algorithm that computes the optimal mechanism
in time $\poly(m,n)$.  Furthermore, there is an optimal
mechanism in which every signal transmitted from the seller
to a buyer is sampled from a set of size $O(m+n)$.
\end{theorem}

In Theorem~\ref{thm:full-surplus}, we prove an
analogue of the 
result of \onlineversion{Cremer and
McLean~\cite{cremermclean88}}{Cremer and
McLean~\citeyear{cremermclean88}}
on optimal auctions with correlated bids.
We show that when the correlation of $\omega$ and $\theta$
is complex enough that a certain matrix has full rank,
the optimal mechanism extracts the full surplus.
However, as in~\cite{cremermclean88}, when this matrix is
ill-conditioned the optimal mechanism can be quite exotic,
using a mixture of
unboundedly large positive and negative payments.
This raises the following question:
To what extent can its revenue be approximated by
simpler and more natural mechanisms?  We explore
this question in Section~\ref{sec:continuity} by
investigating the relative power of four progressively
more general types of mechanisms:
\begin{compactenum}[(i)]
\item
a ``sealed envelope'' mechanism that treats $\omega$
as an indivisible good by writing its value inside a
sealed envelope and posts a price for the 
the envelope;
\item
mechanisms that reveal a signal about $\omega$ but
charge the buyer for this signal {\em before} revealing it;
\item
mechanisms that reveal a signal about $\omega$ and then
charge the buyer a non-negative amount that depends on
the signal;
\item
arbitrary mechanisms.
\end{compactenum}
It is not hard to show that if one compares the optimal
mechanisms from two of these four classes, their revenue
never differs by more a factor of more than $|\Theta|$,
the number of potential buyer types.
Section~\ref{sec:continuity} shows that
this multiplicative gap is tight up to a constant factor:
for any two of the
aforementioned classes of mechanisms, one can find examples where
mechanisms in the more general class obtain
$\Omega(|\Theta|)$ times as much revenue as
the optimal mechanism in the more specific class.

Our work leaves many interesting open problems, we discuss these problems in
Section \ref{sec:open-problems} and point out potential connections to other
areas in computer science (as cryptographic\comment{ and privacy}) and in
economics (as
cheap talk and dynamic mechanism design).


\xhdr{Related work}
The concept of information occupies a notable position in
Auction Theory. In the classic work of
\onlineversion{Milgrom and
Weber \cite{milgrom_weber_82}}{\citeN{milgrom_weber_82}}, 
the authors consider different auction formats for a single item and 
discuss how revenue changes as the seller reveals information (but not directly charging for it) 
regarding the quality of the good. \onlineversion{Persico~\cite{Persico00}}{
\citeN{Persico00}} remarks that the information structure is almost always
assumed to be exogenous and out of the control of the mechanism designer. He
initiates a line of inquire that proposed to endogenize the information
acquisition process in the auction. The information acquisition in his model
occurs by a competing firm paying some fixed exogenous cost (say by
performing R\&D). \comment{ and as a result observing some signal correlated.
\note{MOSHE: correlated with what?}}

It is interesting to consider the qualitative changes when information moves
from an auxiliary device to the position of the central object being sold.
It was noted by many authors that classic results in economics that were
designed for dealing with traditional goods fail when selling information.
\onlineversion{Varian \cite{varian1999markets}}{\citeN{varian1999markets}}
presents such discussion in a much broader context.
Although expressing some similar concerns on the relation with
traditional goods, his definition of information is very different from ours.

Closer to our model are the work of
\onlineversion{Es\"{o} and Szentes \cite{EsoSzentes07}}{\citeN{EsoSzentes07}}
and \onlineversion{H\"{o}rner and Skrzypacz
\cite{HornerSkrzypacz09}}{\citeN{HornerSkrzypacz09}}. In \onlineversion{Es\"{o}
and Szentes \cite{EsoSzentes07}}{\citeN{EsoSzentes07}} the authors develop a
model of consulting, in which the consultant, sells information about a certain
random variable to a client. The assumptions on the nature on the random
variables are orthogonal to ours: theirs are continuous random variables on
$[0,1]$ while ours are discrete random variables. Also, while our utility
function is generic, their utility function is linear. The mechanism obtained
resembles our Pricing Outcomes mechanism, where the payment of the client
depends on the action he takes (in their paper, whether the client decided to
undertake a project or not). Given that their utility function is simpler, the
structure of the disclosed information is also a lot simpler.

The paper of \onlineversion{H\"{o}rner and Skrzypacz
\cite{HornerSkrzypacz09}}{\citeN{HornerSkrzypacz09}} is close to our model of
uncommitted buyers. Their main goal is to design self-enforcing contracts for
an agent to sell a binary High-Low signal about the state of the world to a
firm. The authors remark: ``Lack of commitment creates a hold-up problem: since
the Agent is selling information:, once the Firm learns it, it has no reason to
pay for it''. The model again is orthogonal to ours, since the type of the Firm
(information buyer) is known - so there is no uncertainty from the sellers
perspective. On the other hand, no side has the ability to commit. This ability
in our model leads us to a mechanism design approach to the problem, while the
authors analyze the set of equilibria of a game very much in the spirit of the
cheap talk literature \cite{AumannHart03}.

Our model is also related the treatment of information in \onlineversion{Athey
and Levin \cite{LA01}}{\citeN{LA01}}. The authors model a decision maker who
chooses an action and gets
a reward depending on both the action chosen and an unknown state of the world
from which the decision maker gets an imperfect signal. \comment{They consider,
however,
both the unknown state of the world and the signal acquired by the decision
maker to be real numbers, 
\note{MOSHE: what is the implication of this "real number" assumption. I think it relates to assuming some order, if so we should be explicit about that.}
and they restrict their analysis to a special class of
decision problems called monotone decision problems.}
They assume a total ordering on the space of states of the world and on the
space of signals and restrict their attention to monotone decision processes.
We, on the other hand,
do not make any structural assumptions about the set of possible states of the
world and set of possible signals (except being finite) and also do not assume
anything about
the decision problem. Our approach is also different: while the authors in
\cite{LA01} derive comparative statics for the demand of information, we take
a mechanism design approach.

Also related to our paper is the work of \onlineversion{Admati and
Pfleiderer}{} \cite{AP86,AP90}. Motivated by
financial advice in the market for securities, the authors analyze the
following setting: a monopolist seller has information about how much a share
of a risky asset pays off, which is distributed according to a {\em normal}
random variable. They consider a two step process: in the first round, the
seller is able to sell this information to a continuum of traders using some
mechanism. In the second round, the buyers (traders) trade in a speculative
market. The authors consider the problem of how to design mechanisms to
maximize revenue. Although of a similar spirit, the model of Admati and
Pfleiderer is orthogonal to ours. Their model considers multiple buyers that
trade in the same market in the second round, instead of simply making a
decision. On the other hand, many features of our model are absent in theirs:
they do not consider uncertainty of the seller about the buyer's utility
function and the fact that the buyers might already have some private
information correlated with the seller's information. Moreover, they also limit
the seller's power to transform the signal. In their model, the only
transformation the seller is able to perform on the signal is to add normally
distributed noise to it.

Our notion of {\em uncommitted buyers} is in the spirit of participation
constraints used in the literature on dynamic mechanism design \cite{BV06,AS07}.
This literature considers multi-round interaction between a mechanism and the
agents, and for such mechanisms it is required that individual participation
constraints hold in {\em every} period.
In parallel, when considering uncommitted buyers we require that participation is
voluntary at every point the buyer plays, that is, that the buyer does not
defect throughout the protocol. Note that unlike that literature, our buyer only
gets one exogenous signal, so the different ``periods'' are not points of
getting new exogenous information, but rather points where
he takes an action in a multi-round protocol with the seller.

\section{Setting: Context and Protocols}
\label{sec:model}

We consider a setting with a buyer and a monopolist seller. The buyer is a decision maker
and his decision can be represented as picking an action $a \in A$.
His reward from this action depends on the state of the world which is unobservable.
However, both buyer and seller get private signals about the state of the world.
Let $\theta \in \Theta$ be the private signal of the buyer and $\omega \in
\Omega$ be the private signal of the seller.
In this paper we consider the case that both  $\Theta$  and $\Omega$ are finite.
The buyer's expected reward for
taking action $a$ when the two signals are realized to 
$\theta, \omega$ is given by $u(\theta, \omega,
a)$. The pair of private signals comes from a joint distribution $\mu \in
\Delta(\Omega \times \Theta)$.\footnote{As usual, given any set $X$ we represent the set of
probability distributions over $X$ by $\Delta(X)$. We denote by
$\mu(\omega,\theta)$ the probability of the event $(\omega,\theta)$ and by
$\mu(\omega)=\sum_{\theta} \mu(\omega,\theta)$ and $\mu(\theta)=\sum_{\omega} \mu(\omega,\theta)$ the marginals of $\omega$ and $\theta$
respectively.} We denote the prior on $\omega$ by the vector $p \in \R_+^\Omega$,
i.e., $p(\omega) = \mu(\omega)$. We call the pair $(u,
\mu)$ the \textbf{context} and assume it is common knowledge.
\comment{ among the seller and the buyer.}

To illustrate the model, suppose the buyer is an Internet advertiser who has
acquired one display-ad slot and is deciding which ad to show to the user. So,
the set $A$ represents the possible ads he can place on this slot. The
effectiveness of each ad depends on who the user is exactly. This is unknown to
the buyer, but he has a private signal $\theta$, which is the user browsing
history in the website. Consider now the seller as a data provider who has
information about age, gender, geographic location and income range of
the user. Let this information be encapsulated in a signal~$\omega$.

Another interpretation is to consider $\theta$ as the type of the buyer. Since
the reward of the buyer $u(\theta, \omega, a)$ is a function of $\theta$, one
can use the same model to express the seller's uncertainty about the
buyer's reward function.

The information that the seller holds is valued by the buyer.
If the buyer observe his signal $\theta$ and nothing more,
his expected reward is $\max_a \E_{\omega}[u(\theta, \omega, a) \vert \theta]$,
where the expectation is taken over $\omega$ sampled from $\mu(\cdot\vert \theta)$.
If he also learns the value of $\omega$ exactly, his expected
reward increases to $\E_{\omega}[\max_a u(\theta, \omega, a) \vert \theta]$.
His surplus from knowing $\omega$  is thus:
$$\buyersurplus(\theta) = \E_{\omega}[\textstyle\max_a u(\theta, \omega, a)
\vert
\theta] - \textstyle\max_a \E_{\omega}[u(\theta, \omega, a) \vert \theta]$$
Ideally, the seller would like to extract this extra surplus as revenue, but as
she faces uncertainty
regarding the buyer (she does not know $\theta$) and since the buyer act
strategically,
generically the seller would not be able to extract all that surplus.
The central question of this paper is, ``What mechanisms can the seller use in
order to
extract the largest possible fraction of this surplus?''

\subsection{Sealed Envelope Mechanism}\label{subsec:Sealed-Envelope}

Before we start exploring the space of all possible mechanisms, we present a
very simple (and usually suboptimal) mechanism --- the Sealed Envelope
Mechanism. We do so in order to highlight some basic difficulties in designing
mechanisms for selling information. In this mechanism the seller treats the
information as if it were a regular good.  She
 writes $\omega$ on a piece of paper, puts it inside an envelope and then offers
the
envelope for a fixed price $t$ to the buyer. If the buyer's type is $\theta$,
then his value for the envelope equals to $\buyersurplus(\theta)$, his surplus
from knowing $\omega$,
so the revenue is $t \cdot \prob_{\theta \sim \mu} [\buyersurplus(\theta)\geq
t]$.
Note that the seller is {\em not} using her knowledge of $\omega$ in determining
$t$.
It is easy to optimize $t$ to obtain the best Sealed Envelope Mechanism.

Notice that, after seeing $\omega$, the seller can update her belief
about $\theta$. The seller might try to
change the mechanism in the following way: enclose $\omega$ in an envelope and
sell it for price $t(\omega)$ maximizing
$t \cdot \prob_{\theta \sim \mu} [\buyersurplus(\theta)\geq t \vert \omega]$.
By doing so, the seller leaks information in the prices. Upon observing price $t(\omega)$, the buyer gains
information about $\omega$ even {\em without} buying the envelope.

\subsection{Generic Interactive Protocol}

\comment{
Mechanism design usually starts by invoking the 
celebrated revelation principle \cite{gibbard-73,myerson-79}. One might be
tempted to apply it out-of-the-box to  this setting, but this would ignore many
subtleties of the fact that the good in question is {\em information}.
Many of those subtleties were discussed in the introduction.

Therefore, our first steps are to define a generic interactive protocol between
the buyer and the seller, to recast the revelation principle in terms
of structural properties of the protocol, and then to investigate to what
extent the revelation principle holds for our setting. We assume that the
protocol is designed by the seller with only knowledge of the context $(u,\mu)$.
The protocol prescribes the behavior of the seller for each $\omega$ and we
assume that the  seller always follows this prescription. All of this is known
to the buyer.
}

Our first step is to define a generic interactive protocol between
the buyer and the seller. We assume that the protocol is designed by
the seller with only knowledge of the context $(u,\mu)$. The protocol prescribes
the behavior of the seller for each $\omega$ and we assume that the seller
always follows this prescription. All of this is known to the buyer.

After this we state Myerson's revelation principle for this setting and then
we formally define the stronger notion of \emph{one-round revelation}. Then we
discuss to what extend it is possible to obtain one-round revelation mechanisms
for this setting.

\begin{defn}[Generic interactive protocol]
 A generic interactive protocol for a context $(u,\mu)$ is a finite decision
tree defined on a set of nodes $N$. For each non-leaf node,
let $C(n)$ be the children of node $n$. Each non-leaf node is labeled either as a seller-node, as
buyer-node or a transfer-node.  Furthermore:
\begin{itemize}
 \item each seller node $n$ has a prescription of the seller behavior, which
associates for each $\omega$ a probability distribution over $C(n)$. Formally,
the prescription on node $n$ is a collection of distributions
$\psi^\omega_n \in \Delta(C(n))$, one for each $\omega$.
\item each transfer node $n$ has only one child and has associated with it a
fixed (possibly negative) amount $t(n)$.
\end{itemize}
\end{defn}
In practice we can think of each edge as labeled with a different message.
Starting from the root, if the seller or buyer moves, she or he sends a
message, which is represented by moving down the tree (picking a child).
As the seller pledges to follow the protocol, her behavior (distribution of
messages she sends conditional on $\omega$) is encoded in seller nodes in
advance. On the other hand, the buyer  strategically decides on the messages he
sends at buyer nodes given the protocol tree defined by the seller.  Moving down
the tree from a transfer node $n$ to its child represents a money transfer; the
value $t(n)$ designates how much is transferred from buyer to seller, so a
negative value represents a payment to the buyer.
No action is taken at leaf nodes.
In each such a node the buyer will update his
belief about $\omega$ based on the protocol history.
This belief is called the {\em posterior probability distribution}.

We consider two types of strategies for the buyer: committed and uncommitted
strategies.  A {\em committed strategy} is one where we can trust the buyer
to follow the entire protocol, i.e., when reaching a buyer node, he sends one
of the messages specified by the protocol and when reaching a transfer node, he
sends or receives the specified amount of money. An {\em uncommitted strategy} is one
where the buyer has, on top of that, the option of defecting from the protocol
in each node,  by simply leaving the mechanism (which is formally captured by allowing him to play ''$\bot$'').
More formally:

\begin{defn}[Buyer strategies]
A \textbf{committed strategy} is a collection of distributions
$\phi_n^\theta \in \Delta(C(n))$ for each
buyer-node or transfer-node $n$ and each type $\theta$.

An \textbf{uncommitted strategy} is a collection of
distributions $\phi_n^\theta \in \Delta(C(n) \cup \{\bot\})$ for each
buyer-node or transfer-node $n$ and each type $\theta$.
\end{defn}
At this point, it is instructive to represent the Sealed Envelope Mechanism in
the form of a generic protocol. It is a tree consisting of three interior nodes:
in the root there is a buyer-node with two children corresponding to the
messages: ``Do not accept the offer'' and ``Accept the offer''. The child
corresponding to ``Do not accept the offer'' is a leaf. The other
child is a
transfer-node with the specified amount of $t$. Its only child is a
seller-node. This seller node has a child for each $\omega \in \Omega$ and
the seller prescription for this node is simply $\psi_n^{\omega}(\omega) = 1$.
See Figure~\ref{fig1} for an illustration.

Two other natural selling strategies are important in this work.
\textbf{\em Pricing Mappings} refers to any posted-price mechanism in which the
seller presents a menu of offers each having the following form: for a
fixed amount of money, the buyer obtains the right to observe a random
signal sampled by the seller
from a distribution that depends on the value of $\omega$
in a pre-specified way.  \textbf{\em Pricing Outcomes} refers to a
similar type of posted-price mechanism
with one crucial difference: rather than charging the buyer
a fixed amount of money after he selects an offer from the menu,
the amount that the buyer pays (or receives) is allowed to depend
on the signal that is revealed by the seller.  This gives the
seller the potential to price-discriminate among buyers whose
different types lead to their having different  beliefs
about the conditional distribution of $\omega$,
and therefore different assessments about the expected cost
of accepting a given offer.  Mechanisms that price
mappings or price outputs can easily be represented in the
form of generic interactive protocols, as illustrated in Figure~\ref{fig1}.

It is important to notice that the seller designs the protocol solely based on
the context, and after she designs the protocol its description becomes
common knowledge. This happens before the pair $(\theta,\omega)$ is drawn.
For example, the price at which
the item is offered in the Sealed Envelope Mechanism is hard-coded in the
protocol, so there is no need for the seller to send a message announcing it.

Now we define the utility associated with a committed strategy for a given protocol.
Each committed strategy $\phi$ induces a distribution over the leaves of the tree: sample
$(\theta, \omega) \sim \mu$, then start in the root and use
$\psi_n^\omega$ and $\phi_n^\theta$ to move down the tree until a leaf is
reached. Let $Z$ be the leaf reached. For each leaf
$\ell$ associate $\tau(\ell)$ the sum of the amounts of the transfer-nodes in the
path between $\ell$ and the root. We define the utility of buyer of
type $\theta$ for $\phi$ as:
$$U(\theta, \phi) = \E_{\omega}[ \textstyle\max_a \E[u(\theta,\omega,a) \vert
\theta, Z] - \tau(Z) \vert \theta]$$
We say that a protocol is \textbf{voluntary} is there is a committed strategy
$\phi$ such that \onlineversion{$$U(\theta, \phi) \geq \max_a
\E_{\omega}[u(\theta, \omega, a) \vert
\theta], \quad \forall \theta.$$}{$U(\theta, \phi) \geq \max_a
\E_{\omega}[u(\theta, \omega, a) \vert
\theta]$ for all $\theta$.} This means that the expected utility the buyer gets
from participating in the protocol is at least as large as the utility he would
get by not participating in it.
A committed strategy is called \textbf{optimal} if for all $\theta$ and for all
alternative committed strategies $\phi'$, $U(\theta, \phi) \geq U(\theta, \phi')$.
The \textbf{revenue} extracted by this protocol is: $\E[\tau(Z)]$.

We can define similar concepts for uncommitted strategies: for a given node $n$,
let $\tau(n)$ be the sum of the amounts in the transfer nodes in the path
between the root and $n$ (not including $n$). An uncommitted strategy defines a
distribution $Z$ over the nodes of the tree (not necessarily leaves) and we can
therefore define $Z, U(\theta, \phi)$, optimal strategy and revenue in the exact
same way.
Notice that every protocol is trivially voluntary for uncommitted buyers,
since there is always a strategy guaranteeing
the buyer $\max_a \E_{\omega}[u(\theta, \omega, a) \vert
\theta]$, which is the strategy that defects at any transfer node.

\begin{defn}
 We say that it is possible to {\em extract revenue $R$} from a committed (uncommitted)
buyer in a context $(u,\mu)$ if there is a voluntary protocol for this context
and an optimal committed (uncommitted) strategy $\phi$ for this protocol with
revenue at least $R$.
\end{defn}

Notice that in the case of committed buyers 'voluntary' is an important
restriction, because otherwise, one could simply have a mechanism consisting
solely of a transfer node of amount $R$ and a leaf. For uncommitted buyers,
however, the only optimal strategy in such a mechanisms would be to defect in
the root.

\begin{defn}
We define the {\em optimal revenue} that can be extracted from a committed (uncommitted)
buyer in a context $(u,\mu)$ to be the maximum $R$ such that for any
$\epsilon>0$ it is possible to extract revenue $R-\epsilon$ from a committed
(uncommitted) buyer in a context $(u,\mu)$.
\end{defn}

\begin{figure}
\centering
\includegraphics[scale=.9]{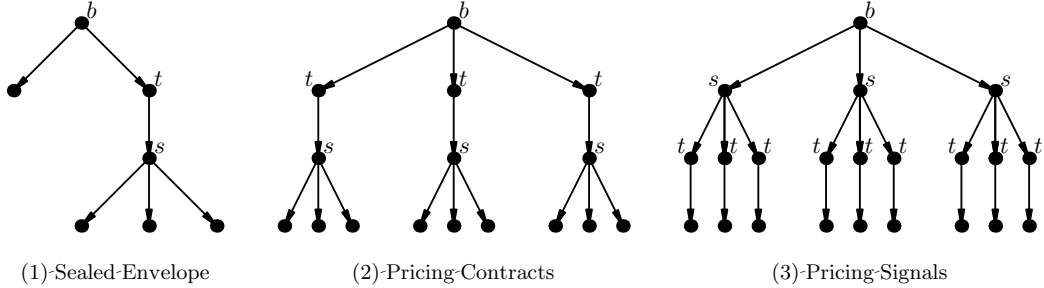}
\caption{Illustration of various protocols. A node labeled by $b$ is a buyer node, a node labeled by $s$ is a seller node, and a node labeled by $t$ is a transfer node.}
\label{fig1}
\end{figure}

\subsection{Revelation Principle and One-Round Revelation Mechanisms}

First we define the concept of a \emph{revelation mechanism}, in the sense of
Gibbard \cite{gibbard-73} and Myerson \cite{myerson-79} and recast their
celebrated revelation principle in our setting (its proof is included, for
completeness, in Appendix \ref{appendix:model}). 
For our purposes it would be sufficient to formulate
it in terms of revenue (traditional formulations are somewhat more general).

\begin{defn}[Revelation Mechanism]
\label{def:revelation mechanism}
A \emph{revelation mechanism} is a protocol represented by a
tree where the root is buyer node, where the buyer is asked to report his
type. Moreover, there are no other buyer nodes in the tree.  A strategy in
such a protocol is \emph{truthful} if the buyer reports his type truthfully in
the root.
\end{defn}

\begin{theorem}[Revelation Principle]\label{thm:uncommitted_partial_revelation}
\label{thm:myerson_revelation}
Consider any context $(u,\mu)$. If it is possible to extract revenue $R$ from a
committed (uncommitted) buyer in the context $(u,\mu)$, then there is a revelation
mechanism and a committed (uncommitted) truthful strategy that is optimal for it and
produces revenue $R$.
\end{theorem}

Theorem \ref{thm:myerson_revelation} says that we can restrict our attention to
Revelation Mechanisms. However, this still allows trees with arbitrary depths
and complicated arrangements of seller and transfer nodes. Next we define the
stronger notion of \emph{One-round revelation mechanisms}, where  only a very
simple interaction between the seller and the buyer is allowed:

\begin{defn}[One-round Revelation Mechanism]
\label{def:one_round_revelation}
A \emph{one-round revelation mechanism} is a revelation mechanism represented
by a tree of depth three where each path from the root has at most one
vertex of each type (i.e. at most one seller node and at most one transfer
node).
\end{defn}

Now we present formal definitions of two special types of One-round Revelation
Mechanisms which were briefly discussed earlier in this section:

\begin{defn}[Pricing Mappings and Pricing Outcomes]
\label{def:contracts-signals}
A \emph{Pricing Mappings Mechanism} is a
truthful direct revelation mechanism in which
all the children of the root are transfer nodes
and all their children are seller nodes.
A \emph{Pricing Outcomes Mechanism} is a
truthful direct revelation mechanism in which
all the children of the root are seller nodes
and all their children are transfer nodes.
\end{defn}

\section{Independent signals}\label{sec:independent}

In this section we analyze the case when $\omega$ and $\theta$ are independent,
which is simpler than the general case.
In this case, the seller's belief about $\theta$ and the buyer's belief about
$\omega$ are common knowledge. First we prove that for this setting, we can
focus on One-round Revelation Mechanism when searching for the optimal
mechanism. Then we show how to compute it efficiently using a convex program. 
Finally we show that there always exists a protocol with a fairly small tree.

\begin{theorem}[Existence of a One-Round
Optimal Mechanism]\label{thm:rev-principle-independent}
Consider any context $(u, \mu)$ such that $\theta$ and $\omega$ are
independent.
If it is possible to extract revenue $R$ from a {\em committed} buyer in the
context $(u,\mu)$, then it is possible to extract revenue $R$ from an {\em
uncommitted} buyer in the same context using a Pricing Mappings Mechanism.
\end{theorem}

The proof is presented in Appendix~\ref{sec:ap-independent}.
A consequence of Theorem \ref{thm:rev-principle-independent} is that with
independent signals the fact that the
buyer is committed does not help the seller to extract more revenue.
In this setting it is possible to extract revenue $R$ from uncommitted buyers if
and only if it is possible to extract revenue $R$ from committed
buyers.\\
\xhdr{Pricing Mappings Mechanism}\label{subsec:princing-contracts}

\noindent Theorem \ref{thm:rev-principle-independent} allow us to focus on
Pricing Mappings
Mechanisms.
An alternative way of describing a Pricing Mappings Mechanism is as a
fixed menu of contracts
$\{(Y_\theta, t_\theta)\}_{\theta\in \Theta}$.
The contract $(Y_\theta, t_\theta)$ is intended for a buyer of type $\theta$ (in
the sense that it would be optimal for such a buyer to choose that contract out
of the menu).
A buyer choosing the contract $(Y_\theta, t_\theta)$ would pay $t_\theta$ and
observe one realization of the random variable $Y_\theta$ that is
correlated with $\omega$, taking values in a finite set $\sspace_\theta$.
We call the elements of $\sspace_\theta$ {\em signals}, since they reveal
to the buyer some information about $\omega$.
By buying the contract $(Y_{\theta'}, t_{\theta'})$ a buyer of type $\theta$ gets
utility\footnote{Since this section talks about the case that $\omega$ and $\theta$ are independent, sampling $\omega\sim \mu(\cdot\mid \theta)$ is the same as sampling $\omega\sim \mu$, and thus in the rest of the section when writing $\E_{\omega}$ we mean $\E_{\omega\sim \mu(\cdot\mid \theta)}$}
$\E_{Y_{\theta'}} [\max_a  \E_{\omega} [u(\theta, \omega, a) \vert Y_{\theta'}]] - t_{\theta'}$.
Now, we discuss how to design the menu of contracts
in order to maximize revenue.

To be more precise, $Y_\theta$ is a random variable that is
produced by the seller using $\omega$ and possibly some random bits $r$ that are
independent of $(\omega,\theta)$. Without loss of generality
we represent $Y_\theta$ by a family
$\{\psi^\omega_\theta\}_{\omega \in \Omega}$ for $\psi^\omega_\theta \in \Delta(\sspace_\theta)$.
In order to sample $Y_\theta$, the seller observes $\omega$ and then samples
the value of $Y_\theta$ according to $\psi^\omega_\theta$.

Without loss of generality, we can call $(Y_\theta, t_\theta)$ the {\em favorite
contract} of a buyer of type $\theta$. In order for such set of contracts to be
{\em valid} we need to make sure that: (1) the protocol is voluntary, i.e., the
utility of a buyer of type $\theta$ by taking contract $(Y_\theta, t_\theta)$ is
at least as high as his utility of not participating in the mechanism and acting using his belief given only $\theta$,
and (2) contract $(Y_\theta, t_\theta)$ is indeed his favorite one, i.e., he would not strictly prefer to
misreport his type and buy a contract $(Y_{\theta'}, t_{\theta'})$ for some $\theta'\neq \theta$.
Property (1) ensures \textbf{individual rationality} (\textbf{IR}) and property
(2) ensures \textbf{incentive compatibility} (\textbf{IC}). Formally:

\begin{defn}\label{dfn:valid_contracts}
 A menu of contracts $\{(Y_\theta, t_\theta)\}_{\theta\in \Theta}$ is {\em
valid} if and only if: \comment{ the following
conditions hold:}
$$ \begin{aligned} & \E_{Y_\theta} [\textstyle\max_a
\E_\omega[u(\theta, \omega, a) \vert Y_\theta]] - t_\theta \geq
\textstyle\max_a
\E_\omega[u(\theta, \omega, a)], \quad \forall \theta, &(IR_\theta) \\
& \E_{Y_\theta} [\textstyle\max_a
\E_\omega[u(\theta, \omega, a) \vert Y_\theta]] - t_\theta \geq  \E_{Y_{\theta'}}
[\textstyle\max_a
\E_\omega[u(\theta, \omega, a) \vert Y_{\theta'}]] - t_{\theta'},  \quad \forall \theta
\neq \theta', &(IC_{\theta,\theta'}) \end{aligned}$$
\end{defn}

Given a valid menu of contracts, its associated revenue is given by
$\sum_{\theta\in \Theta} \mu(\theta) \cdot t_\theta$. This definition implicitly
assumes that whenever the
buyer of type $\theta$ is indifferent between contract $(Y_\theta, t_\theta)$
and not buying anything, i.e. the IR constraint is tight, then he buys contract
$(Y_\theta, t_\theta)$. It also assumes that whenever he is indifferent between
 $(Y_\theta, t_\theta)$ and  $(Y_{\theta'}, t_{\theta'})$, he buys  $(Y_\theta,
t_\theta)$.
This assumption is without loss of generality, since given any menu of contracts
$\{(Y_\theta, t_\theta)\}_{\theta\in \Theta}$ with revenue $R$,
for every $\epsilon > 0$ it is possible to produce a menu $\{(Y'_\theta,
t'_\theta)\}_{\theta\in \Theta}$ with revenue $(1-\epsilon)\cdot R$ such that all
IR and IC inequalities hold strictly. We defer the formal proof of this fact to
Lemma \ref{lemma:strict-preferences}.

Our goal is, for any given context, to design the valid
menu of contracts with largest possible associated revenue. We call it the
optimal menu.

Before starting to optimize the menu, consider a couple of definitions: if
$Y_\theta$ is a variable
taking values in a space $\sspace_\theta$, then for each $s \in
\sspace_\theta$, the $Y_\theta$-posterior associated with $s$ is the
distribution $q \in \Delta(\Omega)$ such that $q(\omega) = \prob(\omega \vert
s)$. We define the value of buyer of type $\theta$ for posterior $q$ as:
$$v_\theta(q) = \textstyle\max_{a \in A} \textstyle\sum_\omega q(\omega)
u(\theta, \omega, a)$$
which is a piecewise-linear convex function $v_\theta : \R^\Omega_+
\rightarrow \R$. Usually $v_\theta(q)$ is defined for $q \in \Delta(\Omega)$,
but sometimes it is used for vectors $q \geq 0$ such that $\sum_\omega
q(\omega) \neq 1$. The reader should note, however that it is a homogeneous
function: $v_\theta(\lambda q) = \lambda v_\theta(q), \forall \lambda > 0$.

We next show that we can represent a signal by a distribution over a finite set
of posteriors (Observation~\ref{obs:feasibility-posteriors}) without repetition
(Observation~\ref{obs:no-duplicated-posteriors}). The proofs of the
observations are immediate, but we include in Appendix \ref{sec:ap-independent}
for completeness.
\comment{
We first show that we can assume w.l.o.g.\ that for two elements $s, s' \in
\sspace_\theta$ the $Y_\theta$-posteriors associated with them are different.
This is encapsulated in the following observation:}

\begin{observation}\label{obs:no-duplicated-posteriors}
 Given a menu $\{(Y_\theta, t_\theta)\}_\theta$, if there are $s,s' \in
\sspace_\theta$ such that the $Y_\theta$-posteriors associated with them are
equal, consider the menu obtained by substituting $Y_\theta$ by
$\tilde{Y}_\theta$, where: $\tilde{Y}_\theta =  s$ for $Y_\theta =
s'$ and $\tilde{Y}_\theta =  Y_\theta$ otherwise.
The obtained menu is also valid and the revenue associated with it is the same
as the one of the original menu.
\end{observation}

Given that the posteriors associated with each signal $s \in \sspace_\theta$
are different,
we can represent $Y_\theta$ by a distribution over a finite set of
posteriors,
i.e., a set $Q$ of posteriors, each $q\in Q$ being of the form $q \in \Delta(\Omega)$, and a probability
$x_\theta(q)$ of each posterior $q\in Q$. The condition that a distribution over
posteriors represents a random variable correlated with $\omega$ is in the
following observation, whose proof is immediate.


\comment{
\begin{observation}\label{obs:feasibility-posteriors}
Fix $\theta\in \Theta$. Let $Y_{\theta}$ be a random variable correlated with $\omega$ with support $\sspace_\theta= [k]$ of size $k$.
$Y_{\theta}$ is represented by $\psi^\omega_\theta \in \Delta(\sspace_\theta)$ and it is sampled by sampling $\psi^\omega_\theta$.
Consider a set $Q = \{q_1, \hdots, q_k\}$ such that $q_i\in \Delta(\Omega)$ for all $i$,
and $x_{\theta} \in \Delta(Q)$.
The pair $(Q,x_{\theta})$ represent a distribution over posteriors of a random variable
$Y_{\theta}$ correlated with $\omega$ if and only if:
$$
\sum_{q_i\in Q} x_{\theta}(q_i) \cdot q_i(\omega) = \mu(\omega| \theta)\ \ \ \  \forall \omega\in \Omega,
$$
In the case that $\theta$ and $\omega$ are independent this is equivalent to
$$
\sum_{q_i\in Q} x_{\theta}(q_i) \cdot q_i(\omega) = \mu(\omega)\ \ \ \  \forall \omega\in \Omega, \qquad (F_\theta)
$$
\end{observation}
}

\begin{observation}\label{obs:feasibility-posteriors}
Let $Y$ be a random variable correlated with $\omega$, i.e., given $\psi^\omega
\in \Delta(\sspace)$, $\sspace=[k]$, $Y$ is sampled by first sampling $\omega \sim
\mu$ and then
sampling from $\psi^\omega$. Consider a set $Q = \{q_1, \hdots, q_k\} \subset
\Delta(\Omega)$ and $x \in \Delta(Q)$.
The pair $(Q,x)$ represent a distribution over posteriors of a random variable
$Y$ correlated with $\omega$ if and only if:
$$\textstyle\sum_{q_i\in Q} x(q_i) \cdot q_i(\omega) = \mu(\omega)\ \ \ \
\forall \omega\in \Omega, \qquad (F)
$$\end{observation}

From this point on, we represent each $Y_\theta$ as a function $x_\theta :
\Delta(\Omega) \rightarrow [0,1]$ with finite support, satisfying
Equation~$(F)$ for each $\theta$.
$$
\textstyle\sum_{q_i\in Q} x_\theta(q_i) \cdot q_i(\omega) = \mu(\omega)\ \ \ \
\forall
\omega\in \Omega, \qquad (F_\theta)
$$
For any finite set of posteriors $Q \subset \Delta(\Omega)$, we can
formulate a restricted revenue maximization problem --- for mechanisms that
offer a menu of contracts with posteriors restricted to belong to $Q$ ---
as $\mathbf{LP_1}$ below
with variables $x_\theta(q)$ for each $\theta \in \Theta$, $q \in Q$
and $t_\theta$ for each $\theta \in \Theta$. Recall that $p \in \Delta(\Omega)$
is the prior on $\omega$, i.e., $p(\omega) = \mu(\omega)$.

$$\begin{aligned} \mathbf{LP_1:} \quad
   & \max \textstyle\sum_\theta \mu(\theta) t_\theta  \text{ s.t. }\\
   & \quad \begin{aligned}
   & \textstyle\sum_q v_\theta(q) x_\theta(q) - t_\theta \geq v_\theta(p),
&\forall
\theta, &\quad (IR_\theta)\\
   & \textstyle\sum_q v_\theta(q) x_\theta(q) - t_\theta \geq \textstyle\sum_q
v_\theta(q)
x_{\theta'}(q) - t_{\theta'}, &\forall
\theta' \neq \theta, &\quad (IC_{\theta,\theta'}) \\
   & \textstyle\sum_{q} x_{\theta}(q) \cdot q(\omega) = \mu(\omega)  &\forall
\theta,
\omega, &\quad (F_{\theta})\\
   & x_\theta(q), t_\theta \geq 0, &\forall \theta, q
  \end{aligned} \end{aligned} $$

The constraints in $\mathbf{LP_1}$ correspond to the characterization of valid
contracts in Definition~\ref{dfn:valid_contracts} and the feasibility of the
representation of contracts as distributions over posteriors in
Observation~\ref{obs:feasibility-posteriors}.

To work with a linear program of finite size,
we have restricted $q$ to belong to a finite set $Q$.
It is conceptually easier
to think of $q$ as ranging over the entire set $\Delta(\Omega)$, in
which case $\mathbf{LP_1}$ would represent the revenue maximization
problem in full generality.  The following lemma, whose proof
is in the appendix, shows that the
restriction to a finite set of posteriors is without loss of
generality.  There is a finite set $Q^*$ that can be precomputed
from knowledge of the function $u$ alone, such that solving
$\mathbf{LP_1}$ with $Q=Q^*$ is guaranteed to produce
an optimal menu of contracts.

\begin{lemma}[Interesting Posteriors]\label{lemma:interesting-posteriors}
 Given $u:\Theta \times \Omega \times A \rightarrow \R$, there is a finite set
$Q^* \subset \Delta(\Omega)$ such that for all $\mu$, the
maximum revenue that can be extracted by any
protocol can also be extracted by a protocol that is limited
to use posteriors in $Q^*$.
Moreover, all the
elements $q \in Q^*$ can be represented with polynomially many bits.
\end{lemma}

By passing to the dual of $\mathbf{LP_1}$,
we get an LP with $O(\abs{\Theta}^2)$ variables.
Below, we give a separation oracle for the dual
showing that it can be solved in polynomial time. The variables of the dual are
$g_\theta, h_{\theta,\theta'} \in \R$ and $y_\theta \in \R^{\Omega}$:

$$\begin{aligned} \mathbf{DLP_1:} \quad
 & \min \textstyle\sum_\theta p^t y_\theta   \text{ s.t. }\\
   &  \begin{aligned}
   & -v_\theta(q) g_\theta + \textstyle\sum_{\theta' \neq \theta} [
v_{\theta'}(q)
h_{\theta',\theta} - v_\theta(q) h_{\theta,\theta'} ] + q^t y_\theta \geq 0, &
\forall \theta, q\\
   & g_\theta + \textstyle\sum_{\theta'\neq\theta} [h_{\theta,\theta'} -
h_{\theta',\theta}] \geq \mu(\theta), & \forall \theta\\
   & g_\theta, h_{\theta, \theta'} \geq 0, y_\theta \in \R^{\Omega}, &\forall
\theta, \theta'
\end{aligned} \end{aligned}$$

\noindent \textbf{Separation oracle}: The second family of constraints is of
size $\abs{\Theta}$, so separating it is trivial. In order to separate the first
family we re-write the constraints in a different way. Notice that the
$v_\theta(\cdot)$ is the maximum over $\abs{A}$ linear functions. So, we can
substitute each constraint of the first family for the following $\abs{A}$
constraints:
$$ \textstyle\sum_{\theta' \neq \theta}  v_{\theta'}(q) h_{\theta',\theta} \geq
-q^t y_\theta +
\left( g_\theta + \textstyle\sum_{\theta' \neq \theta} h_{\theta,\theta'}
\right) \textstyle\sum_\omega u(\theta,\omega, a) q(\omega)$$
for all $a\in A$, $\theta \in \Theta$ and $q \in Q^*$. Now, for fixed
$a, \theta$ we want to check if this constraint is satisfied by all $q \in
Q^*$ or find one $q$ for which the constraint is violated.

Relaxing the requirement $q \in Q^*$ to $q \in \Delta(\Omega)$,
this is equivalent to solving the following convex programming problem:
$$\min_{q \in \Delta(\Omega)}  \textstyle\sum_{\theta' \neq \theta}
v_{\theta'}(q) h_{\theta',\theta} +q^t
y_\theta -
\left( g_\theta + \textstyle\sum_{\theta' \neq \theta} h_{\theta,\theta'}
\right)
\textstyle\sum_\omega u(\theta,\omega, a) q(\omega) $$

\begin{lemma}\label{lemma:convex}
The convex programming problem above can be solved exactly in polynomial time.
Given an optimal solution $q \in \Delta(\Omega)$, there must exist another
optimal solution $q^*$ that belongs to $Q^*$, and we can find such a $q^*$
in polynomial time.
\end{lemma}

The proof (given in the appendix) assumes
that the set $A$ is finite and polynomially
bounded. But even if $A$ has exponential size,
if one is able to solve the problem $\max_a \E_{\omega \sim q} u(\omega, \theta,
a)$ for any posterior $q$, we can still solve the problem.
\comment{
\note{\textbf{RDK}: Unsure what the following sentence means.}
If $A$ is infinite,
corresponding to a general convex function $v_\theta(\cdot)$, then we also can
still prove if the conditions in the lemma hold.}

\xhdr{Comparison with Sealed Envelope}
The sealed envelope mechanism presented in Section
\ref{subsec:Sealed-Envelope} can extract at least $1/\abs{\Theta}$ revenue of
the optimal mechanism, by the following simple observation: if a there is a
voluntary protocol that extracts revenue $R$ from a certain context $(u, \mu)$,
then
there is at least one $\theta$ for which $\frac{R}{\abs{\Theta}} \leq
\mu(\theta) \cdot \buyersurplus(\theta)$, where
$\buyersurplus(\theta) =
\E_\omega[\max_a u(\theta, \omega, a)] - \E_\omega[\max_a u(\theta, \omega,
a)]$ is the maximum surplus that can be extracted from a buyer of type $\theta$.
By setting the price of the envelope to $t = \buyersurplus(\theta)   -
\epsilon$ for some tiny $\epsilon > 0$, the mechanism guarantees revenue at
least $\mu(\theta) \cdot \buyersurplus(\theta)$.
In Example \ref{example:contracts-vs-envelope} we
show this bound is tight by presenting a context where the
sealed envelope mechanism can not extract more then
$\Omega(\frac{1}{\abs{\Theta}})$ of the
optimal mechanism.

\xhdr{Protocols with Small Trees}
We just showed how to compute the revenue optimal protocol in polynomial time
when $\theta$ and $\omega$ are independent. We know that the protocol has
polynomial size, where its size is measured by the number of nodes in the
tree representing the protocol.
Here, we make it more explicit and show that there is a
protocol of size $O(\abs{\Omega} \cdot \abs{\Theta} + \abs{\Theta}^2)$. We show
that by bounding the sum of the support size of the random variables in the
menu of contracts.
This bounds the number of leaves of the One-round Revelation Mechanism
(Theorem~\ref{thm:rev-principle-independent}).
As the number of nodes in the tree representing this protocol is at most 3 times the number of leaves (twice the number of leaves plus one node for the buyer and $\abs{\Theta}$ seller nodes), the bound on the size of the tree follows from the bound on the number of leaves.
The proof of the next theorem is in Appendix~\ref{sec:ap-independent}.

\begin{theorem}\label{thm:quadratic support}
Let $n = \abs{\Theta}$ and $m = \abs{\Omega}$. Denote the
support of a vector $x$ by $\Vert x \Vert_0$.
The program $\mathbf{LP_1}$ has a solution where
$\Vert x_\theta \Vert_0 \leq m+n-1$ for all $\theta$, and
$\sum_\theta \Vert x_\theta\Vert_0 \leq mn + \binom{n}{2}$.
Moreover, there are settings in which support of size
quadratic in $n$ is necessary, even when $m=2$.
\end{theorem}

The fact that the quadratic lower bound holds even
for $m=2$ (Example \ref{example:quadratic-dependence}) is somewhat
counterintuitive: even when the
information being sold is a \emph{single bit},
there are contexts in which revenue maximization
requires using signals that consist of $\Omega(\log n)$
bits.

\section{Correlated signals with committed buyers}
\label{sec:sec:correlated}

In Section~\ref{sec:independent} we have considered independent signals and observed that for that case the seller does not care if the buyer is committed or not. In this section we consider correlated signals and committed buyers, the case of correlated signals and uncommitted buyers will be discussed in Section~\ref{sec:uncommitted}. Throughout this section we assume that the buyer is committed. 

\subsection{Pricing Outcomes Mechanism}\label{sec:correlated}

In the previous section we showed that if $\theta,\omega$ are independent, then
the optimal protocol had the form of offering a menu of contracts with a
fixed price for each contract. In this section we show that this is {\em not}
sufficient to optimize the revenue whenever $\theta,\omega$ are correlated. In
order to optimize revenue, we need to add a twist: we still offer a menu of
options, each option having a random variable $Y_\theta$ correlated with
$\omega$ and taking values in $\sspace_\theta$. Instead of a fixed price,
however, we charge a specific price for each outcome $s \in \sspace_\theta$ of
the signaling scheme.  We continue to refer to the options on the menu
as contracts, although this means that the word \emph{contracts} has
a slightly different meaning in this section than in the preceding one
(as it is no longer the case that the buyer pays before observing any signal).

Why does this construction help in designing mechanisms to optimize revenue?
Suppose a seller designs a variable $Y$ taking values in $\sspace$. Consider
$\psi^\omega \in \Delta(\sspace)$ for each $\omega$ such that the seller
produces $s \in \sspace$ by sampling $s$ according to
$\psi^\omega$. If $\omega$ and $\theta$ were independent, the seller
would be always choosing the same distribution from the buyer's perspective,
which would be $\prob(s) = \sum_\omega \mu(\omega) \cdot \prob(\psi^\omega =
s)$. However, since $\theta$ and $\omega$ are correlated, different buyer-types
perceive different distributions over $\sspace$: a buyer of type $\theta$
perceives $\prob(s \vert \theta) = \sum_\omega \mu(\omega \vert \theta) \cdot
\prob(\psi^\omega = s)$. So, if we condition the prices on the outcomes $s$, two
different buyer-types see different prices for the same contract.  This
increases
our power of price-discrimination.
For the case of {\em committed} buyers we are able to show the existence of an
optimal One-Round Revelation Mechanism.

\begin{theorem}[Existence of a One-Round
Optimal Mechanism]\label{thm:rev-principle-correlated}
For any context $(u, \mu)$,
if it is possible to extract revenue $R$ from a {\em committed} buyer in this
context, then there is a Pricing Outcomes Mechanism that does so.
\end{theorem}
\comment{
\note{MOSHE: It is hard to understand the argument below as we did not explain
the idea of the proof of Theorem 3.1 in the body of the paper.} 
}
The proof, which is given in Appendix~\ref{sec:ap-correlated}, is
essentially the same as the proof of
Theorem~\ref{thm:rev-principle-independent},
except for the last step.  Previously this step relied
on the independence of $\theta$ and $\omega$; here, we
instead rely on the fact that seller nodes 
are 
situated
\emph{above} transfer nodes in our protocol, which eliminates
the need to estimate expected transfers over a random
execution of the original protocol, and instead lets us
match transfers pointwise.

It should be noted that the fact that buyers are committed to follow the
mechanism until the end is crucial.  In fact, in any protocol containing
a transfer node in which the buyer needs to pay the seller and whose child is a leaf,
the optimal uncommitted strategy
would be to defect at that transfer node.
In other words,
uncommitted buyers could acquire the information and leave without paying.
We mention one way to solve this problem: before the mechanism starts, ask a
large sum of money from the buyer. Run the mechanism and then gives the large
sum of money back to the buyer. This will guarantee that the buyer follows
the mechanism until the end, so as not to lose his initial deposit.
This will add an
extra level to the mechanism: one transfer node in the beginning to charge this large
sum of money. The rebate in the end can simply blend with the last transfer.
\comment{
This solution is sometimes implemented by data providers
on the web: charge a reimbursable fee for participating, but retain the
fee in case a buyer does not complete the transaction.}
\comment{
\note{MOSHE: we should be careful here. To motivate the study of uncommitted
buyers in Section 6 we need to state that this solution of rebate is not always
feasible or desirable.}}


\xhdr{Pricing Outcomes Mechanism}\label{subsec:princing-signals}
We will describe a Pricing Outcomes Mechanism using the following
notation.  The seller designs a menu
of contracts solely based on the context.
The menu is a collection $\{(Y_\theta,
t_\theta)\}_{\theta \in \Theta}$ where $Y_\theta$ is a random variable
correlated with $\omega$ and taking values in a finite set $\sspace_\theta$,
just like
in the independent case. The payment function, however, is a function $t_\theta
: \sspace_\theta \rightarrow \R$ and it allows for both positive and negative
payments. The seller outputs $s \sim Y_\theta, t_\theta(s)$, which is a signal
and a payment request of $t_\theta(s)$. If $t_\theta(s) < 0$, the seller
transfers $\abs{t_\theta(s)}$ to the buyer.

In the following, we
discuss how to find the menu maximizing
revenue using a convex program.
The derivation of this convex program closely parallels the
derivation of the corresponding convex program in the
independent case.  We give the full details
in Appendix~\ref{sec:ap-correlated}, and here
we limit ourselves to discussing the two most salient
differences between the derivation of the convex program
in the independent and correlated cases.
\begin{compactenum}[(1)]
\item  The posterior vector after receiving a given signal
depends upon the buyer's type.  To compare posteriors across
different types, we adopt a common ``frame of reference'' ---
that of an outside observer who observes the signal sent by
the seller but does not observe $\theta$ --- and we
translate from the type-$\theta$ reference frame to
the outside observer's reference frame using a matrix
$D_\theta$ that expresses Bayes' rule.
\item  Since payments are now associated with signals, the
obvious way of expressing the expected revenue is as a
sum of products, where each term is the probability of
sending a particular signal, $x_\theta(q)$, multiplied by the amount
charged in the event of sending it, $t_\theta(q)$.  Thus, our program
has a quadratic objective function if we treat both
the probability and the transfers as primal variables.
(This issue does not arise in Pricing Mappings,
since a buyer of type $\theta$ pays the same amount
regardless of what signal is sent, hence the
variables $x_\theta(q)$ do not appear in the objective function.)
To make the objective function linear when pricing outputs,
we define new variables
$\tilde{t}_\theta(q) = x_\theta(q) \cdot t_\theta(q)$.
Fortunately, this change of variables makes
the constraints linear as well.
\end{compactenum}


\xhdr{Full Surplus Extraction}
Correlation can be very valuable to the seller.  In fact, if the
distribution exhibits
sufficiently complex correlation,
the seller might be able to extract full surplus from
the buyers using a Pricing Outcomes Mechanism.
Given a context $(u,\mu)$ we
define the full surplus as the expected gain the buyer would get by learning
the value of $\omega$.  In other words, the full surplus is
$\E_{\theta} [ \buyersurplus(\theta) ] $,
where $\buyersurplus(\theta) = \E[\max_a u(\theta, \omega, a) \vert \theta] -
\max_a \E[u(\theta, \omega, a) \vert \theta]$.
Clearly no mechanism can extract more then the full surplus, and extracting a
$1/\abs{\Theta}$ fraction of it is trivial, even using a sealed envelope
mechanism, as was observed in Section~\ref{sec:independent}.
Now, we show that if $\mu$ is sufficiently correlated, then we can
extract the full surplus.
Our result leverages the ideas developed by \onlineversion{Cremer and
McLean~\cite{cremermclean85,cremermclean88}}{Cremer and
McLean~\citeyear{cremermclean85,cremermclean88}} in their work
on auctions with correlated bidders, although obviously the
setting in which we apply these ideas is different from theirs.

For a joint distribution $\mu$ over $\Omega \times \Theta$ we define
$\text{rank}(\mu)$ as the rank of the $\abs{\Omega} \times \abs{\Theta}$ matrix
defined by $\mu(\omega,\theta)$. For example, if $\omega$ and $\theta$ are
independent, then $\text{rank}(\mu)=1$.


\begin{theorem}\label{thm:full-surplus}
 If $\text{rank}(\mu) = \abs{\Theta}$ then the optimal Pricing Outcomes
Mechanism extracts the full surplus. Moreover, this can be done with a single
contract.
\end{theorem}

\begin{proof}
We define one single contract in the following way. Calculate $t \in \R^\Omega$
 such that:
$$ \textstyle{\sum}_{\omega \in \Omega} t(\omega) \cdot \mu(\omega,\theta) =
\mu(\theta) \cdot \buyersurplus(\theta), \forall \theta$$
Since $\text{rank}(\mu) = \abs{\Theta}$, this system is guaranteed to have a
feasible solution. Now, offer the following contract to all the buyers: the
seller reveals the value of $\omega$ and requests a payment of $t(\omega)$.
By the definition of $t$, each buyer is indifferent between
buying this contract or not buying anything, so the mechanism is voluntary.
\end{proof}

Notice that one can, in the manner of Lemma \ref{lemma:strict-preferences},
offer
the above contract with price $(1-\epsilon)\cdot t(\omega)$ for each outcome,
getting revenue arbitrarily close to the full surplus and making the players
strictly prefer to buy the contract.

At this point it is instructive to consider a concrete example for full surplus
extraction.
\begin{example}\label{example:surplus-correlated}
Imagine a box that has a locker on it and there are two keys labeled with $0$
and $1$, exactly one of which can open the box. The buyer can choose one key and
try
it. If he opens the box, he gets the object inside. Let the type of the
buyer $\theta$ encode his value $z_\theta$ for the object and some signal that
gives him a hint of which might be the right key. The seller knows exactly what
is the correct key, and let this be $\omega$. How should the seller sell the
information to the buyer?

Consider $\Omega = A = \{0,1\}$, $\Theta = [2]$ and the reward function as
$u(\theta, \omega, a) = z_\theta \cdot \one\{\omega = a\}$, with $z_1=3, z_2 =
5$. The joint distribution is $[\mu(\omega,\theta)] =
\left[ \begin{smallmatrix} 0.2 & 0.3 \\ 0.3 & 0.2 \end{smallmatrix} \right]$.

Before participating in the mechanism a buyer of type $\theta=1$ has interim
belief 
$(0.4, 0.6)$ on $\omega$, so his best action is to pick key $1$, getting
expected reward $0.6 \cdot 3 = 1.8$. If he is able to pick key $0$ whenever
$\omega=0$ and key $1$ whenever $\omega=1$, then he always get a reward of $3$,
so his
value for the information is $\buyersurplus(1) = 3-1.8 = 1.2$. Similarly
$\buyersurplus(2) = 2$. In order to design a contract that extracts full
surplus, find $t(\omega)$ such that:
$$
 0.4 \cdot t(0) + 0.6 \cdot t(1) = \buyersurplus(1) \qquad \text{and} \qquad
 0.6 \cdot t(0) + 0.4 \cdot t(1) = \buyersurplus(2)
$$
Solving this system, we get: $t_1(0) = 3.6, t_1(1)= -0.4$. This means that if
the seller reveals signal $\omega=0$, the buyer needs to pay $3.6$, if the
seller reveals $\omega=1$, the buyer receives $0.4$ from the seller.
Both buyer types see the full information contract being offered, but because of
the correlation of $\theta$ and $\omega$, they perceive its expected cost to be
different. Player of type $\theta$ perceives the expected cost to be
$\E[t(\omega) \vert \theta] = t(0) \mu(\omega=0 \vert \theta) + t(1)
\mu(\omega=1 \vert \theta).$
This feature of correlation gives the seller a great power to do price
discrimination.
\end{example}

\subsection{Continuity and Approximation}
\label{sec:continuity}
One might be tempted to conclude from Theorem \ref{thm:full-surplus} that for
any distribution $\mu$ such that $\abs{\Theta} \leq \abs{\Omega}$, one is able
to extract full surplus since all matrices $[\mu(\omega,\theta)]$ can be
approximated arbitrarily closely
by matrices of full rank. The flaw in this argument is
obvious. If one sees $\mu$ as a $\Theta \times \Omega$ matrix, and
$\buyersurplus
\in
\R^\Theta$ as a vector with the surplus $\buyersurplus(\theta)$ in component
$\theta$, the payment vector $t \in \R^\Omega$ in Theorem
\ref{thm:full-surplus} can be found by solving the linear system $\mu \cdot t =
\buyersurplus$. If $\mu$ is a perturbation of, say, a rank-one matrix
(corresponding to $\theta$ and $\omega$ being independent) then the
linear system is very ill-conditioned, and therefore the solution $t$ has a
very high norm. This causes $t$ to diverge as $\mu$ becomes closer to being
independent. To illustrate this point, let us revisit Example
\ref{example:surplus-correlated}.

\begin{example}
 Consider the same reward function as in Example
\ref{example:surplus-correlated} but with a different probability distribution:
$\Omega = A = \{0,1\}$, $\Theta = [2]$,
$u(\theta, \omega, a) = z_\theta \cdot \one\{\omega = a\}$, with $z_1=3, z_2 =
5$, and distribution $[\mu(\omega,\theta)] =
\left[ \begin{smallmatrix} 0.25 & 0.25 \\ 0.25 & 0.25 \end{smallmatrix}
\right]$.

The surplus of each buyer is given by $\buyersurplus(1) = 1.5, \buyersurplus(2)
= 2.5$.
Proceeding as described in Section \ref{sec:independent}, one finds
that the optimal mechanism is to offer one single contract that outputs
$\omega$ exactly and costs $1.5$. Clearly both types buy this contract and the
expected revenue is $1.5$.
Notice however that the expected surplus is $\frac{1}{2} \buyersurplus(1) +
\frac{1}{2} \buyersurplus(2) = 2$ and Theorem \ref{thm:full-surplus} guarantees
that for a slightly perturbed joint distribution one can extract it entirely.
For example, consider
$\mu = \left[ \begin{smallmatrix}
         0.25001 & 0.24999 \\
         0.24999 & 0.25001
        \end{smallmatrix} \right].$
Applying the proof of Theorem \ref{thm:full-surplus} we get the following
mechanism extracting revenue $2$: offer the contract that outputs the full
information and if the outcome is $\omega$, a player pays $t(\omega)$ where
$[t(0), t(1)] = [-24996, 25004]$.
\end{example}

This example highlights two problems with the optimal mechanism.
The most obvious is that it somehow abuses risk-neutrality. The optimal mechanism
produces very large payments which are balanced by large rebates. This
situation is clearly not desirable in practice.
The second problem is that the revenue
that can be extracted from a certain context might change abruptly whenever the
context changes slightly.

It turns out that these discontinuities in the optimal-revenue
function are only unidirectional: as one varies the context,
the revenue can abruptly decrease but it cannot abruptly increase.
Furthermore, for certain restricted classes of mechanisms, the
optimal revenue depends continuously on the context.  In particular,
this holds for the first three members of the
following sequence of progressively more general types of mechanisms.\footnote{
A protocol has {\em no positive transfers} if at any transfer node $n$ the seller never pays the buyer, i.e.\ $t(n) \geq 0$. We discuss Pricing Outcomes Mechanism with No Positive Transfers as this mechanism achieves the optimal revenue for committed buyers with no positive transfers.  (This follows from the proof of Theorem~\ref{thm:rev-principle-correlated}.)

Note that we use the term \emph{no positive transfers} as it is used in literature on auction theory, to exclude mechanisms in which the seller pays money to the buyer. We chose to use this standard term although in our setting such a payment is  \emph{represented} by a transfer node with \emph{negative} transfer $t(n) < 0$.}

{ \small
$$\left\{ \begin{array}{c} \text{Sealed} \\ \text{Envelope} \end{array} \right\}
\subseteq
\left\{ \begin{array}{c} \text{Pricing} \\ \text{Mappings} \end{array} \right\}
\subseteq
\left\{ \begin{array}{c} \text{Pricing Outcomes} \\ t\geq 0
\end{array} \right\}
\subseteq
\left\{ \begin{array}{c} \text{Pricing} \\ \text{Outcomes} \end{array}
\right\}$$
}
To formalize these continuity assertions, let us fix
$\Omega, \Theta,$ and $A$, and regard a context
$(u,\mu)$ as a point in the topological space
$\mathcal{C} = \R^{\Theta\times\Omega\times A} \times \Delta(\Omega \times
\Theta)$ equipped with its standard topology.  The
revenue of the optimal mechanism for committed buyers (or, equivalently,
the optimal Pricing Outcomes Mechanism) will be denoted
by $R(u,\mu)$.  Similarly, we use $R_e,R_c,R_p$ respectively to
denote the revenue of the optimal Sealed Envelope
Mechanism, Pricing Mappings Mechanism, or
Pricing Outcomes Mechanism with No Positive Transfers.
The following theorem formalizes the continuity assertions
claimed above.
Recall that a function $f:\R^k \rightarrow \R$ is {\em lower-semicontinuous}
if for all converging sequences $x_i \rightarrow x$, $f(x) \leq \liminf f(x_i)$.
\begin{theorem} \label{thm:continuity-summary}
The function $R : \mathcal{C} \to \R$ is lower-semicontinuous.
The functions $R_e,R_c : \mathcal{C} \to \R$ are continuous.
Let $\mathcal{C}^{\circ}$ be the set of contexts
$(u,\mu)$ such that $\mu(\omega,\theta) > 0$ for all $\omega,\theta$
(henceforth called \emph{non-degenerated contexts}), the
function $R_p : \mathcal{C}^{\circ} \to \R$ is continuous.
\end{theorem}
Appendix~\ref{subsec:ap-continuity-semicont} breaks
the theorem down into pieces and proves each piece.
The fact that $R$ is lower-semicontinuous
is proven by appealing to Lemma~\ref{lemma:strict-preferences}
which shows that for any context $(u,\mu)$,
the optimal revenue can be approximated arbitrarily
closely by mechanisms in which all of the IR and IC constraints are
slack.  Such a mechanism remains valid in a neighborhood of
$(u,\mu)$, and its revenue varies continuously in this neighborhood,
hence the optimal revenue throughout a neighborhood of $(u,\mu)$
remains nearly as great as the optimal revenue at $(u,\mu)$.
As a side effect of the method of proof, we obtain a
robustness result: if the seller believes that the context $(u, \mu)$ and
designs an protocol extracting $R(u, \mu)$, but the real context may actually
be slightly different, 
the seller can design a slightly different protocol extracting
revenue at least $(1-\epsilon) \cdot R(u, \mu)$
from any context $(u', \mu')$ that is sufficiently close to $(u, \mu)$.
We note that a proof in the same spirit for the Cremer and McLean setting can
be found in \onlineversion{Robert \cite{robert91}}{Robert
\citeyear{robert91}}. Notice that the heart of our proof is the
application of Lemma \ref{lemma:strict-preferences}. In this specific
ingredient, our proof is different from that of Robert.

The same argument establishes lower-semicontinuity of
the functions $R_e,R_c,R_p$.
Their upper-semicontinuity follows from a compactness
argument.  For each of these classes of mechanisms, the optimal
revenue can be found by solving a variant of $\mathbf{LP_2}$ that
reflects the additional constraints on the mechanism.  Each such
variant LP has an unbounded feasible region, and the first
step of the proof is to identify a compact (i.e.\ closed and
bounded) subset of the feasible region that is guaranteed to
contain the optimal LP solution, at least for
all contexts in a neighborhood of a given $(u,\mu)$.
Taking any sequence of contexts $(u_i,\mu_i)$ converging
to $(u,\mu)$ and passing to a subsequence if necessary,
we may assume that their optimal mechanisms $\mathcal{M}_i$
converge to a limit $\mathcal{M}$ belonging to
the feasible region of the LP.  This mechanism $\mathcal{M}$
then supplies a lower bound on the optimal revenue in
context $(u,\mu)$ that suffices to prove upper-semicontinuity.

\xhdr{Approximating Revenue}
We have seen that the Sealed Envelope Mechanism achieves
at least a $(1/n)$-approximation to the optimal revenue,
where $n=|\Theta|$,
and that this bound cannot be improved by more than a constant
factor in the worst case.  The other two classes of mechanisms
listed above --- Pricing Mappings, and Pricing Outcomes without
positive transfers --- are substantially simpler and more
natural than Pricing Outcomes in full generality, so it would
be desirable to approximate the optimal revenue using one
of these simpler classes of mechanisms.  Unfortunately, in the
worst case, the approximation achieved is rather poor: 
we prove in
Appendix~\ref{sec:ap-continuity-approx} that
there exist contexts $(u,\mu)$ such that
$ R_c(u,\mu) \leq R_p(u,\mu) \leq O(1/n) \cdot R(u,\mu).$
The proof is
an application of Theorem~\ref{thm:continuity-summary}.
We carefully construct a context with independent $\omega,\theta$ for which the
full surplus exceeds $R_c(u, \mu)$ by a factor of  $\Omega(n)$.
For such a context Pricing Mappings is optimal, i.e. $R_c(u, \mu) = R_p(u, \mu)
= R(u, \mu)$.
Now if we slightly
perturb $\mu$ to make it into a full-rank matrix,
Theorem~\ref{thm:full-surplus} tells us that $R(u,\mu)$
jumps up by a factor of $\Omega(n)$ to match the full surplus,
while Theorem~\ref{thm:continuity-summary} ensures that
$R_c(u,\mu)$ and $R_p(u,\mu)$ can only change by a tiny
amount.

Finally there remains the question of whether
$R_c$ achieves a good approximation to $R_p$ in
the worst case.  As it happens, once again the
worst-case ratio between these two quantities is
$\Omega(n)$, as evidenced by Example~\ref{ex:rcrp} in the Appendix. 
In this example the seller gets from a side channel the information about the
the buyer's surplus. In a Pricing
Outcomes Mechanism (even without positive transfers), the seller is capable of
leveraging this information, while in a Pricing Mappings Mechanism, she is not.

\section{Correlated Signals with Uncommitted Buyers}\label{sec:uncommitted}

In this section we examine mechanisms for uncommitted buyers
when signals are correlated, and we
formulate an intriguing open question related to a surprising failure of the
one-round revelation property. First, we review what is known about uncommitted
buyers from previous sections. Theorem \ref{thm:rev-principle-independent} says
that if $\omega$ and $\theta$ are independent, the revenue-optimal mechanism is
aPricing Mappings Mechanism.
The optimal strategy for both committed and uncommitted buyers is the same in
such a mechanism. However, when $\omega$ and $\theta$ are correlated,
the revenue-optimal mechanism for committed buyers is a Pricing Outcomes
Mechanism (Theorem \ref{thm:rev-principle-correlated}).
In such a mechanism, the buyer declares his type, the seller sends a signal and
the buyer pays a certain amount of money that depends on the signal sent by the
seller. Such a mechanism clearly does not work for uncommitted buyers, who can
defect after getting the
signal but before the payment.

In section \ref{sec:correlated}, we mentioned one way to get around this problem:
before executing the Pricing Outcomes Mechanism the seller can require the buyer
to deposit a large amount of money, then the mechanism executes, and after it
completes the seller refunds the buyer.
This mechanism has clear practical drawbacks:
in practice a deposit-refund scheme increases the cost of participation as it
requires the buyer to always be able to make large payments, even in  cases in
which at the end of the protocol the net payment is small, or no payments are
made in the execution.
The latter case is particularly problematic as it is usually costly to establish
a payment relationship (e.g. the buyer needs to spend time giving his credit
card information) so always imposing payments might deter some buyers.

A protocol has \textbf{no positive transfers} if at any transfer node money
always goes from the buyer to the seller, i.e.\ $t(n) \geq 0$ for every
transfer node $n$.
Once we exclude positive transfers the problem of designing optimal mechanisms
for uncommitted buyers becomes quite challenging.  We formulate it as the
following open problem:

\begin{open_problem}\label{open_problem:uncommitted}
Characterize the protocols that extract maximum revenue from uncommitted buyers
subject to  no positive transfers.  In particular, is it possible to design an
algorithm that decides, given any context
$(u,\mu)$ and target revenue $R$, whether there exists a protocol with no
positive transfers that extracts revenue $R$ from uncommitted
buyers?
\end{open_problem}

A natural first attempt to address this problem would be to prove, in the
manner of Theorem \ref{thm:rev-principle-independent} and Theorem
\ref{thm:rev-principle-correlated}, that for this setting one can extract
optimal revenue using a One-Round Revelation Mechanism. We show that this
approach fails.

\begin{theorem}[Failure of the One-round
Revelation Property]\label{thm:revelation_failure}
There exists a context $(u,\mu)$ with correlated $\omega$
and $\theta$ for which some generic protocol with no positive transfers extracts
strictly more revenue from uncommitted buyers than {\em any} One-round
Revelation Mechanism with no positive transfers.
\end{theorem}

In order to prove this theorem, we observe that if it is possible to
extract revenue $R$ from an uncommitted buyer using a One-round Revelation
Mechanism, then it is possible to extract the same revenue using a Pricing
Mappings Mechanism. This follows easily from the fact that if a transfer node
$n$ is the last node before the leaf in a path of the protocol tree, an
uncommitted buyer will always defect before this node if $t(n) > 0$.

Then in Example~\ref{example:uncommitted_separation} we present a context for
which an interactive protocol with no positive transfers extracts strictly more
revenue than any Pricing Mappings Mechanism.

\begin{example}\label{example:uncommitted_separation}
We present a context $(u,\mu)$ for which a mechanism where the seller interacts
with the buyer twice (producing a protocol tree of height $4$) extracts
strictly more revenue than any direct revelation mechanism.

Let $\Omega = \{0,1\}$, $\Theta = \{0,1\}$ and the distribution
$\mu = \left[\begin{smallmatrix} 0.3 & 0.2 \\ 0.2 & 0.3  \end{smallmatrix} \right]$.
Let $A = \{0,1\}$ and define the utility such that $u(\theta,\omega,0) = \omega$
(for $a=0$) and $u(\theta,\omega,1) = 1-9\omega$ (for $a=1$).
As usual, since $\Omega = \{0,1\}$ we represent the posterior probability by one real number $q \in [0,1]$, the probability of the event $\omega = 1$.
The first part of Figure  \ref{fig:uncommitted_separation} depicts the utility as a function of the posterior $q$, the utility is
represented by functions $v_\theta :[0,1]\rightarrow \R$.  In this particular case $v_0(\cdot) = v_1(\cdot) = v(\cdot)$.

It is simple to see that the optimal Pricing Mappings Mechanism offers a single contract,
pricing the full information (value of $\omega$) at $0.4$, and getting revenue of $0.4$.
\comment{In the notation of Section \ref{sec:continuity}: $R_e(u,\mu) = R_c(u,\mu) = 0.4 < 0.6 = R_p(u,\mu) = R(u,\mu)$.}

\begin{figure}
\centering \includegraphics[scale=0.9]{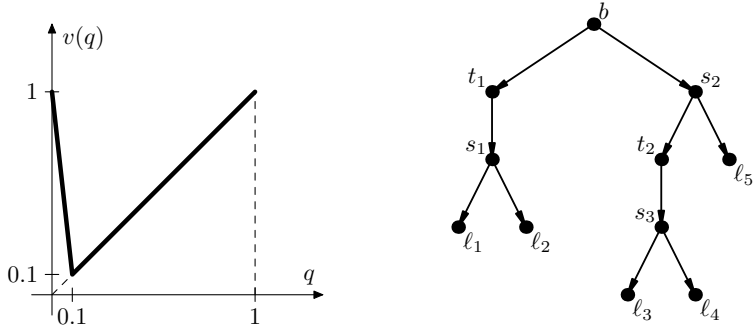}
\caption{Example \ref{example:uncommitted_separation}: a protocol for uncommitted
buyers generating strictly more revenue then any direct-revelation-mechanism.}
\label{fig:uncommitted_separation}
\end{figure}

We next present a protocol with no positive transfers that extracts strictly more revenue
than the optimal Pricing Mappings Mechanism from uncommitted buyers.
The protocol is represented by a tree of height $4$ depicted in the second part of Figure \ref{fig:uncommitted_separation}.
It consists of two transfer nodes with amounts $t_1 = 0.533$ and $t_2 = 0.8$. For
the seller nodes, the transition probabilities are as follows:
\begin{itemize}
 \item \textbf{node $\mathbf{s_1}$:} \comment{the seller outputs the full
information,
i.e., }the seller outputs $\ell_1$ whenever $\omega=0$ and outputs $\ell_2$
whenever $\omega=1$.
 \item \textbf{node $\mathbf{s_3}$:} \comment{is also a full-information node.
}the seller
outputs $\ell_3$ whenever $\omega=0$ and outputs $\ell_4$
whenever $\omega=1$.
  \item \textbf{node $\mathbf{s_2}$:} the seller moves either to node
$t_2$ or to node $\ell_5$ according to the following probabilities:
$\prob(\ell_5 \vert\omega=0)= 0$, $\prob(\ell_5 \vert \omega=1) = 5/6$ and
clearly $\prob(t_2 \vert \omega) = 1-\prob(\ell_5 \vert \omega)$.
\end{itemize}
Now, we claim that the optimal strategy for an uncommitted buyer is to
play left (to node $t_1$) whenever $\theta=0$ and
play right (to node $s_2$)
whenever $\theta=1$, and then follow the protocol (make the transfers when asked) without defecting.
We prove this in Claim~\ref{lem:uncommitted_separation} in
Appendix~\ref{sec:ap-uncommitted}.
Given these strategies we calculate the expected revenue of the protocol, which is:
$\mu(\theta=0) \cdot 0.533 + \mu(\theta=1) \cdot 0.8 \cdot \prob(t_2 \vert
\theta=1) = \comment{\tfrac{1}{2} \cdot 0.533 + \tfrac{1}{2} \cdot 0.8 \cdot
\tfrac{1}{2} = } 0.4665 > 0.4$
\end{example}

We have shown above that generic protocols can extract {\em strictly more}
revenue than the One-round Revelation Mechanism when buyers are
uncommitted and no positive transfers are allowed. One might wonder if such
interactive mechanisms can extract as much revenue from uncommitted buyers as
can be extracted from committed buyers. We show that the answer is no, as for
the setting of Example~\ref{example:uncommitted_separation} there is a gap
between the two.

\begin{theorem}\label{thm:gap-committed-uncommitted}
There exists a context with correlated $\omega$ and $\theta$
for which the optimal revenue that can be extracted from {\em uncommitted} buyers
using a protocol with no positive transfers is strictly less than the optimal
revenue that can be extracted from {\em committed} buyers.
\end{theorem}
We prove the theorem in Appendix~\ref{sec:ap-uncommitted}.
For context $(u,\mu)$, the optimal revenue that can be extracted from committed buyers using protocols that has no positive transfers is exactly the revenue $R_p(u,\mu)$ that can be extracted by
{\em Pricing Outcomes with No Positive Transfers Mechanism}.
To prove the claim we show in Appendix~\ref{sec:ap-uncommitted} that
for the setting of Example~\ref{example:uncommitted_separation}
for some $\delta>0$ it is impossible to extract revenue of $R_p(u,\mu)-\delta$
from uncommitted buyers using protocols that have no positive transfers.

The intuition behind the proof of Theorem \ref{thm:gap-committed-uncommitted} is the
following: in this context, it is possible to extract the full surplus from
committed buyers. In order to extract close to this much revenue from uncommitted
buyers, the mechanism must be offering an option that results with some posterior very
close to the full information. Now we can show that since a buyer of $\theta=1$
is paying at most his entire surplus, there is a deviation for a buyer
of type $\theta=0$ that guarantees him almost full information for a price
considerably below his surplus.

\onlineversion{
\section{Open Problems}\label{sec:open-problems}

We believe the design of mechanisms for selling information is an area full of
exciting possibilities. In particular,
there might be
potential connections with information theory \comment{, privacy} and
cryptography, that
could be
discovered and formalized. In this section we describe some
interesting open directions, listing them from the more precise to the more
vague problems:

\begin{enumerate}
\item \textbf{Mechanisms for Uncommitted Buyers:} The main open problem left
by our paper is the problem of designing revenue optimal protocols for
uncommitted buyers, when no positive transfers are allowed (Open Problem
\ref{open_problem:uncommitted}).
Theorem \ref{thm:revelation_failure}
reveals that
multiple rounds of partial information disclosure (interleaved by payment to the
seller) are sometimes necessary to achieve optimal revenue if the buyer is
allowed to abort his interaction with the seller prematurely.
Solving this open problem would probably
require the development of new tools for directly bounding
the revenue of interactive mechanisms
--- in sharp contrast to most optimal mechanisms in the
literature which satisfy some version of the One-round Revelation Property.

Possibly, tools from dynamic mechanism design \cite{AS07, Pavan08} and
from cheap talk \cite{AumannHart03} might be useful in dealing with this
question, since both involve repeated interactions between parties before an
outcome is effectively implemented. Nevertheless, our setting differs from them
in fundamental ways: the reason interaction is necessary here is different from
that in dynamic mechanism design. In our problem, information is just revealed
by nature in the beginning of time. No exogenous signals are revealed during
the interaction between seller and buyer. This fact brings our model closer to
cheap talk, yet there is a substantial difference as in our setting talking is
not ``cheap'' in the sense that payments are required so that the conversation
continues.

\item \textbf{Continuous Type Spaces and Structured Contexts:} Our results fix
a finite discrete type space and use tools from Linear and Convex Programming
to design optimal mechanisms. Is it possible to obtain a general theory of
``Selling Information'' where the signals from nature come from arbitrary spaces? Es\H{o} and Szentes \cite{EsoSzentes04} are able to deal with real-valued signals
but they severely restrict the utility function of the agents.

Is there a natural yet less restrictive condition on the context that makes the
problem tractable for more generic spaces?
Can one come up with a general condition that would be still be sufficient for
the  optimal mechanism to have explicit and natural representation (similar to
Myerson's mechanism), rather than being a solution to some mathematical program?

\item \textbf{Multiple buyers:} In our work, we consider only two agents: a
seller and a buyer. In the context of information advertisement, there are
usually many sellers (many information-providing agencies) and many buyers
(advertisers). Those buyers do not simply solve a decision problem once they
acquire the information, but play a game among each other.

A natural next step is to consider a variant of our model where there is a
single seller but multiple buyers, who play a game after acquiring information.
One needs to be careful when defining such problem since we must be sure that
the game has an unique equilibrium or that for each possible outcome of the
selling information phase, the Bayesian game played in the second phase has a
focal equilibrium on which we can concentrate when reasoning about the first
phase. Along those lines is the work of Es\H{o} and Szentes \cite{EsoSzentes07} and
the recent paper by Kempe, Tardos and Syrgkanis~\cite{SKP12}.

\item \textbf{Coupling Goods and Information:} As mentioned earlier, there are
many situations where goods and information are coupled together (for example,
the restaurant example in the introduction). What is the correct model in which
to analyze such situations?

\item \textbf{Dynamic selling of information:} We analyzed a
single interaction between the buyer and the seller. Alternatively one could
consider an ongoing (possibly interactive) relation between those two parties,
e.g. when the buyer is interested in information regarding multiple impressions.
Can the seller benefit from this and extract strictly more revenue from
the whole process than she would if running the optimal auction for each query
individually? It is known that for traditional goods, selling bundles might
generate strictly more revenue than selling goods individually. What is the
form of the analogous results for information?

\item \textbf{Computationally bounded agents:} We assumed that the seller sends
messages to the buyer coming from a certain distribution and the buyer uses
Bayes' rule to learn about the state of the world from those messages. We do not
impose computational constraints on the agents. If we were to assume
computationally-constrained buyers,
the seller could potentially take advantage of cryptographic
primitives, for example sending an encrypted piece of information
to the buyer and later selling the decryption key.
How would such capabilities affect the results?
Would the seller be able to
exploit them in a profitable way?
\comment{
\item \textbf{Privacy:} The ``selling information'' problem raises obvious
privacy concerns if the information being sold is user data.
The notion of \emph{Differential Privacy} \cite{DMNS06} captures what it means
for a data release to preserve privacy, a recent successful line of research
uses this notion to produce mechanisms that are privacy preserving. 
One might like to consider the question of precisely defining what {\em privacy}
means in the context of information selling, and to consider the revenue
maximization problem under privacy constraints. 
}
\end{enumerate}
}{
\section{Open Problems}\label{sec:open-problems}
The design of mechanisms for selling information is an area full of exciting
possibilities. In particular, we believe that there might be potential
connections with information theory\comment{, privacy} and cryptography.
The main concrete open problem left by our paper is the one of designing revenue
optimal protocols for uncommitted buyers, when no positive transfers are allowed
(Open Problem \ref{open_problem:uncommitted}). One might also consider extending
the model to include multiple buyers that are buying information that is useful
when they compete in a second stage (say an auction).  Other interesting
directions are extending the model to computationally bounded agents
(cryptography might prove useful), the dynamic selling of information when the
seller and buyer interact repeatedly, and the connections to privacy. In the
full version of the paper we elaborate on these and other extensions. 
}

%
\onlineversion{
\bibliographystyle{abbrv}
}{
\bibliographystyle{acmsmall}
}
\bibliography{sigproc-full}


\onlineversion{
\appendix
}{
\elecappendix

\medskip
}
\section{Proofs omitted from Section~\ref{sec:model}}\label{appendix:model}

\onlineversion{
\begin{proofof}{Theorem \ref{thm:myerson_revelation}}
}{
\begin{proof}[of Theorem~\ref{thm:rev-principle-correlated}]
}
The idea of the proof is simple: we add in the root a buyer node where  the
buyer is asked to report his type $\theta$. For each branch, there is a copy of
the original protocol where the seller simulates the behavior of the buyer.

Formally, we fix a context and a consider a protocol represented by a tree,
distributions $\{\psi^\omega_n\}_{\omega \in \Omega}$ for each seller node $n$
and transfers $t(n)$ on each transfer node. Now, consider
a set of moves for the buyer $\{\phi_n^\theta\}_{\theta \in \Theta}$ for each
buyer and transfer node. We can assume that for any buyer node $n$,
$\phi_n^\theta$ outputs $\bot$ with zero probability, since any other strategy
is weakly  dominated by a strategy in which
the buyer does not defect at the buyer node and instead defect at the first
transfer node encountered. Let $T$ represent this protocol. 

Now, define the protocol $T_\theta$ as the same tree with all buyer nodes $n$
substituted by seller nodes where the seller moves according to $\phi_n^\theta$
for all $\omega \in \Omega$. Now, the tree $T_\theta$ consists only of transfer
and seller nodes.

Now, design a new protocol with a buyer node in the root with
$\abs{\Theta}$ outgoing edges, one corresponding to each type. Attach to the
$\theta$-branch a copy of $T_\theta$. Now, it is simple to see that an optimal
committed (uncommitted) strategy is to report the true type and then follow the
protocol (defecting in the nodes of $T_\theta$ corresponding to the nodes in
$T$ where he defected in the original protocol). If any deviation is profitable
now, the corresponding deviation would be profitable in the original protocol.
Also, the revenue of the truthful strategy in the new protocol is clearly the
same as the revenue of the original strategy in the original protocol.
\onlineversion{
\end{proofof}
}{
\end{proof}
}

\section{Proofs omitted from Section~\ref{sec:independent}}
\label{sec:ap-independent}


\onlineversion{
\begin{proofof}{Theorem~\ref{thm:rev-principle-independent}}
}{
\begin{proof}[of Theorem~\ref{thm:rev-principle-independent}]
}
Fix the context, consider a voluntary protocol and an optimal strategy $\phi$
for a committed buyer achieving revenue $R$.
We show that we can reduce it to a protocol of the specified form,  achieving the
same revenue for uncommitted buyers. Let $Z(\omega, \theta)$ be the distribution over leaves obtained
when we start from the root and use $\psi^\omega_n, \phi^\theta_n$ to move down
the tree. Also, let $L$ be the set of leaves of the tree.

In the root of the reduced protocol, we place a buyer node with
$\abs{\Theta}$ edges out of it. This
node corresponds to the buyer being asked to report his type.
The children of the root are transfer nodes, each corresponding
to a distinct type $\theta$.
The amount at such a transfer node corresponds
to the expected payment of type $\theta$ in the mechanism, i.e.,
$\E_\omega[\tau(Z(\omega,\theta))]$.

Now, consider the node $n$ who is the child of the transfer-node in the
$\theta$-branch. Add a copy of $L$ as its children and for each $\omega$, have
the seller use the distribution $Z(\omega, \theta)$.

Now, we need to check that the new protocol is indeed equivalent to the
original one. Consider the uncommitted strategy where the buyers report their true
types in the first node and then pay according to the transfer node.
This strategy generates the same utility and payments as the ones in the original protocol, so one only needs to check that it is
indeed an optimal strategy for the buyer.

It suffices to see that a buyer of
type $\theta$ would not rather declare $\theta'$ instead or defect in the
middle of the reduced protocol.
Defecting in the middle is clearly not beneficial, since
the protocol is voluntary and the buyer does not learn anything until the last
move.
Also, if declaring
$\theta'$ were profitable for a buyer of
type $\theta$,
he would prefer to
play according to $\phi^{\theta'}$ in the original protocol. This happens for
two reasons.  First,
it generates the same distribution over leaves given $\omega$ as
declaring $\theta'$ in the new protocol would. So the value of the deviations
is the same.  Second,
it generates the same expected payment as declaring $\theta'$ in
the new protocol would.  This depends crucially on independence.  A
buyer playing $\phi^{\theta'}$ experiences the same expected payment
regardless of his type, since $\omega$ is drawn from the
same distribution (independent of $\theta$) and we thus
visit the same nodes with the same probabilities in all cases.
\onlineversion{
\end{proofof}
}{
\end{proof}
}

Notice that the last step in the proof clearly does not hold if
$\omega,\theta$ are correlated, since a buyer of type $\theta$ could change his
declaration to $\theta'$, but he can not prevent $\omega$ to be sampled with
probability $\mu(\omega \vert \theta)$.

\onlineversion{
\begin{proofof}{Observation \ref{obs:no-duplicated-posteriors}}
}{
\begin{proof}[of Observation \ref{obs:no-duplicated-posteriors}]
}
Simply notice that the utilities of each player for each contract in the
original and new menu are the same.
\onlineversion{
\end{proofof}
}{
\end{proof}
}

\onlineversion{
\begin{proofof}{Observation \ref{obs:feasibility-posteriors}}
}{
\begin{proof}[of Observation \ref{obs:feasibility-posteriors}]
}
 Let $q_i \in \Delta(\Omega)$ be the posterior associated with $i$, i.e.
$q_i(\omega) = \prob(\omega \vert Y=i)$ and let $x(q_i) = \prob(Y=i)$. Then:
$$ \sum_{q_i\in Q} x(q_i) \cdot q_i(\omega) = \sum_{i} \prob(Y=i) \cdot
\prob(\omega \vert Y=i) = \mu(\omega) $$
Conversely, given $(Q,x)$ satisfying~$(F)$ we can define a variable $Y$ taking
values in $[k]$ and set its joint distribution with $\omega$ to be
$\prob(\omega,Y=i) = x(q_i) \cdot q_i(\omega)$.
\onlineversion{
\end{proofof}
}{
\end{proof}
}

\onlineversion{
\begin{proofof}{Lemma~\ref{lemma:interesting-posteriors}}
}{
\begin{proof}[of Lemma~\ref{lemma:interesting-posteriors}]
}
The function $u$ defines piecewise linear functions $v_\theta :
\Delta(\Omega) \rightarrow \R$. Each of them induces a partition of
$\Delta(\Omega)$ in $\abs{A}$ polytopes in which $v_\theta(\cdot)$ is linear.
One can combine the $\abs{\Theta}$ partitions, by taking the coarser partition
of $\Delta(\Omega)$ that is simultaneously a subpartition of the one induced by
$v_\theta$ for all $\theta$. This way we obtain a finite partition such that
for all its regions all the $v_\theta$ functions are linear.

Now, let $Q^*$ be the set of vertices of the regions in this partition. Given any
posterior $q \in \Delta(\Omega)$, if it is in region $R$ of the partition, one
can write $q = \sum_i \gamma_i q_i$ where $q_i \in Q^*$ are vertices of region
$R$ and $\sum_i \gamma_i = 1, \gamma_i \geq 0$. Now, given a primal solution to
$\mathbf{LP_1}$, if $x_\theta(q) > 0$ and $q \notin Q^*$, simply increase
$x_\theta(q_i)$
by $\gamma_i x_\theta(q)$ and decrease $x_\theta(q)$ to zero. By the linearity
of $v_\theta(\cdot), v_{\theta'}(\cdot)$ in $R$, it is clear that the resulting
solution is still feasible and has the same objective. By repeating
this process as many times as needed, one ends with a solution where the
support is in $Q^*$.
\onlineversion{
\end{proofof}
}{
\end{proof}
}

\onlineversion{
\begin{proofof}{Lemma~\ref{lemma:convex}}
}{
\begin{proof}[of Lemma~\ref{lemma:convex}]
}
Given a convex programing problem $\min f(x) \text{ s.t. } Ax \leq b$, in order
to show that it can be solved exactly in polynomial time we need to show three
things (see section 5.3 of  \cite{ben2001lectures}):

\begin{itemize}
 \item that the function $f$ can be computed efficiently for each point $x_0$ in
the domain.
 \item that for each point $x_0$ in the domain it is possible to calculate a
subgradient $v \in \partial f(x_0)$. A subgradient is a vector $v \in \R^n$ such
that $f(x)-f(x_0) \geq v^t(x-x_0)$
 \item that there is an optimal solution that can be expressed with
polynomially-many bits.
\end{itemize}

Point $1$ above is trivial. For point number $2$ we use the linearity of the
subgradient: if $v \in \partial g(x_0)$ and $w \in \partial h(x_0)$ then $v+w
\in \partial (g+h)(x_0)$. And since the subgradient of linear functions is
trivial, we now need to show how to compute the subgradient of
$v_{\theta'}(\cdot)$. Notice that $v_{\theta'}(\cdot)$ is the maximum of linear
functions. We use the fact that if $g(x) = \max_j g_j(x)$, $g(x_0) = g_i(x_0)$
and $v_i \in \partial g_i(x_0)$, then $v_i \in \partial g(x_0)$, since:
$$g(x) - g(x_0) = g(x) - g_i(x_0) \geq g_i(x) - g_i(x_0) \geq v_i^t(x - x_0)$$

For point $3$, Lemma \ref{lemma:interesting-posteriors} says that
$\mathbf{LP_1}$
can be written with finitely many variables. Therefore, the dual LP
($\mathbf{DLP_1}$) can be
written with finitely many constraints. Now, the optimizer of the dual can be
found by taking $O(\abs{\Theta}^2)$ constraints,
making them tight and solving the resulting linear system. This solution has polynomially many bits.

Given an optimal solution $q \in \Delta(\Omega)$, for each
$\theta \in \Theta$ let $a(\theta)$ be a solution to
$\max_a \E_{\omega \sim q} u(\omega, \theta, a)$.  There is
a polytope $\mathcal{P} \subseteq \Delta(\Omega)$ that consists of
all posteriors $q'$ such that actions
$(a(\theta))_{\theta \in \Theta}$ remain optimal when
the posterior is $q'$.  The objective function of
the convex program is linear when restricted to $\mathcal{P}$,
so the set of optimal solutions includes at least one extreme point
of $\mathcal{P}$, and such an extreme point can be found in polynomial
time.  Recalling the construction of the set $Q^*$ in the
proof of Lemma~\ref{lemma:interesting-posteriors}, we see
that all of the extreme points of $\mathcal{P}$ are elements
of $Q^*$, so we have established that
our convex program has an optimal solution $q^*$ that belongs to $Q^*$,
and that we can find such a $q^*$ in polynomial time, as claimed
in the lemma.
\comment{
we show that the dual LP can be written with finitely many
linear constraints, each represented by polynomially-many bits. If we are able
to show that, the optimizer can be obtained by taking $O(\abs{\Theta}^2)$
constraints, making them tight and solving the linear system. The solution
clearly has polynomially-many bits.

In order to show that the LP can be expressed with finitely many constraints,
partition $\Delta(\Omega)$ in regions where in each region, all the
$v_\theta(\cdot)$ functions are linear. Now, there are finitely many of those
regions and each of them is a polytope with finite number of vertices that can
be represented with polynomially-many bits. Let $\tilde{\Delta}$
be such a region. Now, notice that we can express that:
$$-v_\theta(q) g_\theta + \sum_{\theta' \neq \theta} [  v_{\theta'}(q)
h_{\theta',\theta} - v_\theta(q) h_{\theta,\theta'} ] + q^t y_\theta \geq 0$$
for each $q \in \tilde{\Delta}$ simply by enforcing it on the vertices of
region $\tilde{\Delta}$. All the other points follow automatically by
linearity, since all the $v_\theta(\cdot)$ functions are linear in
$\tilde{\Delta}$.
}
\onlineversion{
\end{proofof}
}{
\end{proof}
}

The following example, discussed in Section~\ref{subsec:princing-contracts},
reveals that the Sealed Envelope Mechanism does not extract more than
$\Omega(\frac{1}{|\Theta|})$ of the revenue of the optimal mechanism.
\begin{example}\label{example:contracts-vs-envelope}
 Consider $\Omega = \{0,1\}$, $\Theta = [n]$ and
$\mu(\omega, \theta) = \frac{2^{-\theta-1}}{1-2^{-n}}$.
Now, for each $\theta$ we can represent $u(\theta, \omega, a)$ as a
piecewise-linear function $v_\theta(\cdot)$ on $[0,1]$ where $$v_\theta(q) =
\max_a [ u(\theta, 1, a) \cdot q + u(\theta, 0, a) \cdot (1-q) ]$$
Let $v_\theta(\cdot)$ be the function interpolating the following three points:
$(0,0), (1-\frac{1}{T^\theta},0), (1,2^\theta)$, where $T$ is some fixed large
number. Now, the value of a buyer of type $\theta$ for the envelope is
$2^\theta$ and occurs with probability $\mu(\theta) = O(2^{-\theta})$, so any
price $p$ will only be able to extract $O(1)$ revenue.

Now, we show a Pricing Mappings Mechanism that can extract $\Omega(n)$
revenue. Consider the menu where the price of contract $\theta$ is $t_\theta =
(1-\frac{1}{T}) 2^\theta - \epsilon$ and outputs the posterior
$1-\frac{1}{T^{\theta+1}}$ w.p. $\frac{1}{2}$ and the posterior
$\frac{1}{T^{\theta+1}}$ w.p. $\frac{1}{2}$. This is the same as outputting a
variable $Y_\theta$ with joint distribution:
$$[\prob(Y_\theta,\omega)] = \begin{bmatrix}
                                 1/(2 T^{\theta+1}) & 1/(2 (1-T^{\theta+1})) \\
                                 1/(2 (1-T^{\theta+1})) & 1/(2 T^{\theta+1}) \\
                                \end{bmatrix}$$
It is simple to check that this is a valid menu and that it generated revenue
$\Omega((1-\frac{1}{T})n)$.
\end{example}

\onlineversion{
\begin{proofof}{Theorem~\ref{thm:quadratic support}}
}{
\begin{proof}[of Theorem~\ref{thm:quadratic support}]
}
Given any context, fix an optimal solution $\{x_\theta(q), t_\theta\}_\theta$
for $\mathbf{LP_1}$.
We can assume $\Theta = [n]$ and that the solution is such
that $t_1 \geq t_2 \geq t_3 \geq \hdots \geq t_n$.  Let $Q_\theta \subset
\Delta(\Omega)$ be the posteriors in the support of $x_\theta$. We can restrict
our attention to $\mathbf{LP_1}$  with variables $x_\theta(q)$ for $q \in
\cup_\theta Q_\theta$.

For any given $\theta \in [n]$ we show that we can substitute $x_\theta(\cdot)$
by some $x'_\theta(\cdot)$ with support at most $\theta+m-1$. In order to do so
fix all $t_{\theta'}$ and $x_{\theta'}(\cdot)$ for all $\theta' \neq
\theta$.
Now find $x'_\theta(q),t'_\theta$ solving the following system:

$$\begin{aligned} & \max t'_\theta \text{ s.t. } \\
& \quad \begin{aligned}
   & \sum_{q\in Q} v_\theta(q) x'_\theta(q)  - t'_\theta \geq
\max\{v_\theta(p),\max_{\theta' \neq \theta} \sum_{q \in Q} v_\theta(q)
x_{\theta'}(q) - t_{\theta'} \}, &\quad(IR_{\theta})+(IC_{\theta,\theta'},
\theta' \neq \theta) \\
   & \sum_{q\in Q} v_{\theta'}(q) x_{\theta'}(q) - t_{\theta'} \geq \sum_{q\in
Q} v_\theta(q)
x'_{\theta}(q) - t_{\theta},\ \ \  \forall \theta' < \theta,
&\quad(IC_{\theta',\theta}, \theta' < \theta) \\
   & \sum_{q\in Q} x'_\theta(q) \cdot q(\omega) = \mu(\omega), \ \ \  \forall \omega
\in \Omega &\quad(F_\theta) \\
   & x'_{\theta}(q) \geq 0, \ \ \  \forall q \in Q
\end{aligned} \end{aligned}$$

The first constraint represents the IR$_\theta$ combined with
$IC_{\theta,\theta'}$ for all $\theta' \neq \theta$. Since all
$x_{\theta'}(\cdot), t_{\theta'}$ are fixed, the rhs of the inequality is
constant. The second family of constraints refer to the IC$_{\theta',\theta}$
constraints for $\theta' < \theta$.

The system of inequalities has $\theta + m$ constraints, so a basic
feasible solution has support at most $\theta + m$, which corresponds to
$t'_\theta$ and $x'_\theta(\cdot)$ such that $\Vert x'_\theta(\cdot)\Vert_0
\leq \theta+m-1$. Consider the solution we
get from substituting $x_\theta(\cdot)$ by $x'_\theta(\cdot)$. The only
constraints in $\mathbf{LP_1}$ which can be violated are:
$$\sum_q v_{\theta'}(q) x_{\theta'}(q) - t_{\theta'} \geq  \sum_q v_{\theta}(q)
x_{\theta}(q) - t_{\theta}$$
for some $\theta' > \theta$.
If this happens for some $\theta'$, substitute
$(x_{\theta'}, t_{\theta'})$ by $(x'_{\theta}, t_{\theta})$. It is simple to
see that the solution obtained is feasible and that the revenue did not decrease.

Finally, observe that the theorem follows as
$$\sum_\theta \Vert x_\theta\Vert_0 \leq \sum_\theta (m-1+\theta) = mn +
\binom{n}{2}$$

We show in Example \ref{example:quadratic-dependence} that the quadratic
dependence in
$\abs{\Theta}$ is necessary if one wants to obtain the optimal revenue.
\onlineversion{
\end{proofof}
}{
\end{proof}
}

\begin{example}\label{example:quadratic-dependence}
Consider $\Omega = \{0,1\}$, $\Theta = [n]$ and as
in Example \ref{example:contracts-vs-envelope} represent $u$ by
piecewise-linear convex functions $v_\theta : [0,1] \rightarrow \R$. For some
small enough parameter $\delta > 0$, define $v_\theta$ as the function that
interpolates points
$$ (0,0), (1-\delta \theta, 0), (1-\delta (\theta-1), \delta^{2\theta -2}),
(1-\delta (\theta-2), \delta^{2\theta
-1}),..., (1-\delta, \delta^\theta), (1, \delta^{\theta-1}) $$
For example, $v_1$ is the function interpolating $(0,0), (1-\delta, 0), (1,1)$.
Function $v_2$ interpolates $(0,0), (1-2\delta,0), (1-\delta, \delta^2), (1,
\delta)$ and so on.

For this example we represent all the posteriors $q \in \Delta(\Omega)$ by a
real number $q(1)$, the posterior probability that $\omega=1$.

Now, define $\mu$ such that $\mu(\omega=0) = \mu(\omega=1) = \frac{1}{2}$.
Before we define the distribution over $\Theta$ we note two things about the
primal LP:
\begin{enumerate}
 \item there is an optimal solution to $\mathbf{LP_1}$ where all
posteriors are in $Q = \{0, 1-n\delta, 1-(n-1)\delta, \hdots, 1-\delta, 1\}$ by
the argument in the proof of the Interesting Posteriors Lemma (Lemma
\ref{lemma:interesting-posteriors}).\comment{
In fact, suppose there is an optimal solution where $q$ is in the support of
$x_\theta$ but $q \notin Q$. Then consider $q_1, q_2 \in Q$ such that $q$ is
between $q_1$ and $q_2$. One can write $q = \gamma q_1 + (1-\gamma) q_2$ for
some $\gamma$. Then we can construct a $x'_\theta(\cdot)$ such that
$x'_\theta(q) = 0$, $x'_\theta(q_1) = x_\theta(q_1) + \gamma x_\theta(q)$,
$x'_\theta(q_2) = x_\theta(q_2) + (1-\gamma) x_\theta(q)$ and $x_\theta(\cdot)
= x'_\theta(\cdot)$ for all other posteriors. Now, by substituting
$x_\theta(\cdot)$ by $x'_\theta(\cdot)$ we still get a feasible solution with
the same revenue. This happens since in the interval $[q_1, q_2]$ all
$v_\theta$ functions are linear, so utilities are unchanged when passing from
$x_\theta$ to $x'_\theta$.}
 \item we can write the primal only with the variables $x_\theta(q)$ for $q
\in Q$. This generates a finite LP. Further notice that the polytope that
defines the feasible set is define by $v_\theta(\cdot)$ and the
distribution $\mu(\omega)$, both already defined. The distribution
$\mu(\theta)$ just appears in the objective function.
\end{enumerate}

Since the feasible region
is already defined and it is a polytope (since by restricting to $Q$ we have a
finite number of constraints), let $V$ be the set of all vertices of this
polytope. Define $\tilde{V} \subset \R_+$ as the set of all $t$ such that there
is $(x_{\theta'}, t_{\theta'})_{\theta' \in \Theta} \in V$ and $\theta$ such
that $t = t_\theta$. In other words, it is the set of all prices appearing in
the vertices of the feasible regions. Given that, pick some $\epsilon$ such
that:

$$\epsilon < \frac{1}{n} \cdot \frac{\min_{t,t' \in \tilde{V}, t \neq t'}
\abs{t-t'}}{\max_{t,t' \in \tilde{V}} \abs{t-t'}}$$
and define $\mu(\theta) = \epsilon^\theta \cdot
\frac{1-\epsilon}{\epsilon(1-\epsilon^n)}$. Now the context is define, we can
analyze the optimal menu of contracts for this context. We will show that the
optimal menu is unique and $\sum_\theta \Vert x_\theta \Vert_0 = \Theta(n^2)$.

We claim that for this value of $\epsilon$, the optimal solution to the primal
LP can be found by optimizing $t_1$ over the feasible set, fixing it, then
optimizing for $t_2$ (having $x_1, t_1$ fixed), and so on. The definition of
$\epsilon$ guarantees that no possible gain in  $t_\theta+1, \hdots, t_n$ can
justify a decrease in $t_\theta$.

For example, for $\theta=1$, if we want to find $\max t_1$ with respect to the
feasible set of the primal LP, the unique solution is to take $t_1$ to be the
surplus $\frac{1}{2}[v_1(0) + v_1(1)] - v_1(\frac{1}{2}) = \frac{1}{2}
v_1(1)$, since $v_\theta(0) = v_\theta(\frac{1}{2}) = 0$ by construction,
and
$x_1(0) = x_1(1) = \frac{1}{2}$. Notice that buyer of type $1$ gets no utility
from this contract (the seller is able to extract full surplus). Now, after
fixing $(x_1, t_1)$ we need to guarantee that for all the other contracts buyer
with $\theta = 1$ is not able to extract any utility from them.

We claim that there is an unique optimal solution and the support of $x_\theta$
is $0, 1-(\theta-1)\delta, 1-(\theta-2)\delta, \hdots, 1$. Moreover, in the
optimal solution, the utility a buyer of type $\theta$ extracts from his own
contract is zero, i.e., $\sum_q x_\theta(q) v_\theta (q) - t_\theta =
v_\theta(p)$.
We show this by induction on $\theta$: in order to find the
optimal $(x_\theta, t_\theta)$ one needs to solve the following program (we are
using the fact that the players $\theta' < \theta$ get zero utility from their
own contracts):

$$
\begin{aligned}
& \max t_\theta \text{ s.t. } \\
& \quad \begin{aligned}
& \sum_{q \in Q} v_\theta(q) x_\theta(q) - t_\theta \geq 0, , & (IR_\theta)\\
& \sum_{q \in Q} v_{\theta'}(q) x_\theta(q) - t_{\theta} \leq 0, \ \  \forall
\theta' < \theta, &(IC_{\theta'}, \theta' < \theta)  \\
& \sum_{q \in Q} (q, 1-q) \cdot x_\theta(q) = (1/2, 1/2), &(F_\theta)\\
&  x_\theta(q) \geq 0,
\end{aligned} \end{aligned}
$$

Substituting the values of $v_\theta(\cdot)$, we get:

$$
\begin{aligned}
& \sum_{j=0}^{\theta-1} \delta^{\theta+j-1} x_\theta(1-j\delta)  \geq
t_\theta\\
& \sum_{j=0}^{\theta'-1} \delta^{\theta'+j-1} x_\theta(1-j\delta) \leq
t_\theta, \forall
\theta' =1, ... ,\theta-1\\
\end{aligned} $$

Now we claim that in the solution maximizing $t_\theta$ all the inequalities
above must be tight. Since it is a triangular system, it guarantee that the
support of $x_\theta$ needs to be of size $\theta+1$.

Consider an optimal $(x_\theta, t_\theta)$ maximizing $t_\theta$ and suppose
that one of the constraints above is not tight. Clearly if the IR constraint is
not tight, we can simply increase $t_\theta$, which is an absurd since the
solution is optimal.

Now, suppose that one of the IC constraints are not tight.
Then let $\theta'$ be the largest $\theta'$ such that $IC_{\theta'}$ is not
tight. Now, it must be the case that $x_\theta(1-\theta'\delta) > 0$. In order
to see that, consider two cases (i) $\theta' < \theta-1$ : this would
imply that constraint $IC_{\theta'+1}$ is also slack, violating the maximality
of $\theta'$ and (ii) $\theta' = \theta-1$ : this would cause the
solution to be infeasible due to a conflict between $IC_{\theta-1}$ and $IR$.

Given that, one can construct a new solution by decreasing
$x_\theta(1-\theta'\delta)$ by a small amount $\rho$ and increasing
$x_\theta(1-(\theta'-1)\delta)$ by the amount $\delta \rho$. This keeps IR and
IC constraints still valid, but violates $\sum_{q \in Q} (q, 1-q) \cdot
x_\theta(q) = (1/2, 1/2)$. It is simple to see one can increase $x_\theta(0)$
and $x_\theta(1-(\theta-1)\delta)$ by suitable values to restore this
constraint. By doing so, we do not violate any of the IC constraints and we make
the IR constraint a little bit more slack, which allows us to increase
$t_\theta$. Since it was optimal from beginning, then all IC constraints must
be tight.

Now, since all IC constraints and IR constraints are tight, and together they
form a triangular system, it is easy to see the solution must have support
$\theta+1$. Now, summing for all $\theta$, we get $\sum_\theta \Vert x_\theta
\Vert_0 = \sum_\theta 1+\theta$ exactly matching the bound in theorem
\ref{thm:quadratic support} for $m=2$.
\end{example}

\section{Proofs omitted from Section~\ref{sec:correlated}}
\label{sec:ap-correlated}

\begin{proof}[of Theorem~\ref{thm:rev-principle-correlated}]
Using the same notation of Theorem \ref{thm:rev-principle-independent}, we
can mimic that proof, but whenever constructing the new protocol, we add a
seller node just after the buyer node with a transition to a copy of $L$ using
$Z(\omega, \theta)$ in the $\theta$-branch. Now, for each child of the seller
node, if it corresponds to a leaf $\ell$ in the original protocol, add a
transfer node with amount $\tau(\ell)$.

Again, the buyer committed strategy where each buyer reports his type truthfully
generates the same outcomes and payments as the original protocol. We only need
to argue that it is optimal. The argument is very similar to the one in
Theorem \ref{thm:rev-principle-independent}, however the inversion of the buyer
and seller nodes allows us to prove point 2 without independence. We notice
that if a buyer of type $\theta$ plays using $\theta'$ he has the same
distribution over leaves as he would get if he played $\phi^{\theta'}$ in the
original protocol - by the new construction, it also means the same
distribution over transfers.
\end{proof}

\xhdr{Computing the optimal Pricing Outcomes Mechanism}
Consider a buyer of type $\theta$ that buys a
contract $(Y,t)$. Let $Y$ be represented by a family $\psi^\omega \in
\Delta(\sspace)$ for each $\omega \in \Omega$.
If a signal $s \in \sspace$ is observed by a buyer with type $\theta$
and also by an outside observer, the two observers have different beliefs
on $\omega$ --- the buyer's belief is $\mu(\cdot | \theta)$ while
the outside observer's is $\mu(\cdot)$ --- so Bayes' rule dictates
their respective posteriors as follows.
$$q^s_\theta (\omega) = \prob(\omega \vert \theta, s) =  \frac{\prob(\omega,
s, \theta)}{\prob(s,\theta)} = \frac{\prob(\omega,
s) \cdot \mu(\theta \vert \omega)}{\prob(s) \cdot \prob(\theta \vert s)},
\qquad
q^s(\omega) = \prob(\omega \vert s) = \frac{\prob(\omega, s)}{\prob(s)}
 $$
So, up to rescaling there is a very simple linear transformation that relates
the posteriors of the buyer of type $\theta$ for a signal and the posterior of
the outside observer. Let $D_\theta$ be an $\abs{\Omega} \times \abs{\Omega}$
diagonal matrix with $\mu ( \theta \vert \omega)$ in the diagonal. Then, we can
write: $$q^s_\theta = \frac{D_\theta \cdot q^s}{\prob(\theta \vert s)} \qquad
\qquad (BU)$$
Given the 1-1 mapping between posteriors from the perspective of buyer of type
$\theta$ and posteriors with respect to an outside observer, from this point
on, whenever referring to posteriors, we always refer to posteriors with
respect to the outside observer and apply the suitable transformation when
necessary.

\xhdr{Representation of contracts} Analogously to what was done in
the independent case, we represent $Y$ by a distribution over
(outside observer) posteriors.
We can use Observation~\ref{obs:no-duplicated-posteriors} to see that it is possible to assume w.l.o.g.\ that
the posteriors associated with each $s \in \sspace$ are different. So, we
can represent $Y_\theta$ as a distribution over the set of posteriors
$\Delta(\Omega)$. More formally, we represent $Y$ by $x : \Delta(\Omega)
\rightarrow \R_+$ with finite support satisfying the feasibility condition in
Observation \ref{obs:feasibility-posteriors}.
Similarly $t$ can be represented as
$t : \Delta(\Omega) \rightarrow \R$, simply by associating each signal with its
posterior. Given that, we can
represent the valuation that a buyer of type
$\theta$ has
for contract $(Y,t)$ as:
$$\begin{aligned}
& \E_Y [ \max_a \E[u(\omega,\theta, a) \vert \theta, Y]] =
\sum_{s \in \sspace} \prob(Y=s \vert \theta) \cdot v_\theta\left( \frac{D_\theta
q^s}{\prob(\theta \vert s)} \right) =
\sum_{s \in \sspace} \frac{\prob(s \vert \theta)}{\prob(\theta \vert s)} \cdot
v_\theta(D_\theta q^s) =
\\
& \quad
= \sum_{s \in \sspace} \frac{\prob(s)}{\prob(\theta)} v_\theta(D_\theta q^s)
= \frac{1}{\prob(\theta)} \cdot \sum_{q}
v_\theta(D_\theta q) \cdot x(q)
  \end{aligned}
 $$
where the first line uses the fact that $v_\theta$ is homogeneous
of degree 1, and the transition from the first to the second line
makes use of Bayes' rule in the form
$\prob(s \vert \theta) / \prob(\theta \vert s) =
\prob(s) / \prob(\theta)$.
The buyer's expected payment for this contract is given by:
$$\begin{aligned}
&\E_Y[t(Y=s) \vert \theta
] = \sum_{s \in \sspace} t(s) \cdot
\frac{\sum_\omega \prob(s, \omega, \theta)}{\prob(\theta)} =
\frac{1}{\prob(\theta)} \sum_{s \in \sspace} t(s) x(q^s) \sum_\omega q^s(\omega)
\mu(\theta \vert \omega) = \\ & \quad = \frac{1}{\prob(\theta)} \sum_{q} t(q)
x(q) \cdot (\one^t D_\theta q)&
\end{aligned}$$
where $\one$ is the vector in $\R^\Omega$ with all $1$ coordinates.\\

\noindent \textbf{Linear Programming Formulation:} Now, in order to optimize the
revenue from the protocol we need to design a menu of contracts that is
\emph{valid}, in the sense that player $\theta$ indeed chooses contract
$(Y_\theta, t_\theta)$ and this choice is voluntary. One can define this in a
similar manner to Definition \ref{dfn:valid_contracts} making the appropriate
changes. Since the definition and follow up discussion very closely mirrors the
one in the independent case, we proceed to designing the LP.  As in the independent case, there is an ``interesting posteriors lemma'' that allows us to limit our attention to contracts whose posterior vectors belong to a predefined finite set $Q^*$.  This time the set $Q^*$ depends not only on $u$ but on $\mu$, since the operator $D_\theta$ depends on $\mu$.

\begin{lemma}[Interesting Posteriors for correlated signals]
Given $u:\Theta \times \Omega \times A \rightarrow \R$ and $\mu \in \Delta(\Theta \times \Omega)$, there is a finite set
$Q^* \subset \Delta(\Omega)$ such that for all $\mu$, the
maximum revenue that can be extracted by any
protocol can also be extracted by a protocol that is limited
to use posteriors in $Q^*$.
Moreover, all the
elements $q \in Q^*$ can be represented with polynomially many bits.
\end{lemma}
\begin{proof}[sketch]
The proof exactly parallels the proof of Lemma~\ref{lemma:interesting-posteriors}, except that now we must use the linear transformation $D_\theta$
to translate every buyer's piecewise-linear valuation function
$v_\theta$ into the outside observer's frame of reference.  In
other words, we partition $\Delta(\Omega)$ into polyhedral
pieces on which each of the function $q \mapsto v_\theta(D_\theta q)$
is linear, and let $Q^*$ be the set of vertices of the regions in
this partition.
\end{proof}

Let us continue deriving the LP formulation of the revenue maximization
problem.  The variables in the LP are
$(x_\theta(q), \tilde{t}_\theta(q))$ for all
$\theta \in \Theta, q \in Q^*$. We define $\tilde{t}_\theta(q) =
t_\theta(q) \cdot x_\theta(q)$. We write a linear program
that expresses maximization of revenue subject to
the analogues of the usual IR, IC and feasibility constraints for the
correlated case.

$$\begin{aligned}
& \mathbf{LP_2:} \max \sum_\theta \sum_q (\one^t D_\theta q) \cdot
\tilde{t}_\theta \text{ s.t. }\\
& \quad \begin{aligned}
& \sum_q v_\theta(D_\theta q) x_\theta(q) - (\one^t D_\theta q)
\tilde{t}_\theta (q) \geq v_\theta(D_\theta p), & \forall \theta,
&\quad(IR_\theta)
\\
&  \sum_q v_\theta(D_\theta q) x_\theta(q) - (\one^t D_\theta q)
\tilde{t}_\theta (q)  \geq \sum_q v_\theta(D_\theta q)
x_{\theta'}(q) - (\one^t D_\theta q) \tilde{t}_{\theta'} (q), & \forall \theta'
\neq \theta, &\quad(IC_{\theta,\theta'}) \\
& \sum_q x_\theta(q) \cdot q(\omega) = \mu(\omega), & \forall \theta,\omega,
&\quad(F_\theta) \\
& x_\theta(q) \geq 0,  &\forall q, \theta
\end{aligned} \end{aligned}$$

\begin{remark}
Before solving $\mathbf{LP_2}$ to find the optimal mechanism, it is important
to remark that the variables of the LP correspond to the posteriors as seen an
{\em outside observer}, i.e., someone that knows $\mu$ but only observes the
signal that the seller outputs and not the value of $\theta$. Equation $(BU)$
shows how to calculate the posteriors of the buyer given the posterior of the
outside observer. It is instructive to consider the feasibility constraint
$(F_\theta)$ in $\mathbf{LP_2}$ when viewed from the perspective of the buyer.
First, we can re-write $(F_\theta)$ as:
$\sum_s \prob(s) \cdot q^s(\omega) = \mu(\omega) $, multiplying by $\mu(\theta
\vert \omega)$ and substituting by the definition of $q^s_\theta$ given by
the Bayesian update equation $(BU)$ we have:
$$\sum_s \prob(s) \cdot q^s_\theta(\omega) \cdot \prob (\theta \vert s) =
\mu(\omega, \theta)$$
Now, by manipulating the probabilities, we get:
$$\sum_s \prob(s\vert \theta) \cdot q^s_\theta(\omega) =
\mu(\omega \vert \theta)$$
which shows that the feasibility constraints from the perspective of the buyers
is simply the feasibility constraints from the perspective of the observer
transformed by the Bayesian update $(BU)$ operator.
\end{remark}

Solving $\mathbf{LP_2}$ can be done using exactly the
same techniques used in the
independent case: passing to the dual and using convex programming to design a
separation oracle. We omit the details on how to solve it, since it essentially
mirrors the independent case.

Once we have a solution
$x_\theta(q), \tilde{t}_\theta(q)$ with support on $\hat{Q} \subseteq Q^*$, we
want to recover a solution $(x_\theta, t_\theta)$.
Clearly one can obtain $t_\theta(q)$ by taking $\tilde{t}_\theta(q) /
x_\theta(q)$ whenever $x_\theta(q)$ is positive.  What if it is equal to zero?
In order to solve this problem we take the following steps: (1) we show that
there is one feasible solution with revenue
at least $(1-\epsilon)$ of the optimal, where all the IC and IR constraints are
slack (2) then we increase $x_\theta(q)$ slightly for all $q \in \hat{Q}$ and
renormalize so that feasibility holds and (3) we safely calculate
$t_\theta(q) = \tilde{t}_\theta(q) / x_\theta(q)$, since now we are sure that
$x_\theta(q)$ is positive. The following lemmas show how to perform steps (1)
and (2):

\begin{lemma}[Strict Preferences]\label{lemma:strict-preferences}
Given an optimal solution $(x_\theta,
\tilde{t}_\theta)$ to $\mathbf{LP_2}$, then there is a solution $(x'_\theta,
\tilde{t}'_\theta)$ whose objective is at least a $(1-\epsilon)$-fraction of the
optimal, such that all IC and IR constraints are slack, except
possibly $IC_{\theta,\theta'}$ when
$(x'_\theta,\tilde{t}'_\theta) = (x'_{\theta'}, \tilde{t}'_{\theta'})$,
i.e., when the contracts for $\theta$ and $\theta'$ are exactly the same.
\end{lemma}

\begin{proof}
 First, suppose a certain IC constraint is tight, say:
 $$\sum_q v_\theta(D_\theta q) x_\theta(q) - (\one^t D_\theta
q) \tilde{t}_\theta (q) = \sum_q v_\theta(D_\theta q)
x_{\theta'}(q) - (\one^t D_\theta q) \tilde{t}_{\theta'} (q). $$
Then it must be the case that $(\one^t D_\theta q)
\tilde{t}_\theta (q) \geq (\one^t D_\theta q) \tilde{t}_{\theta'} (q)$.
Otherwise one could replace $(x_\theta, \tilde{t}_\theta)$ with $(x_{\theta'},
\tilde{t}_{\theta'})$ as buyer $\theta$'s contract and improve revenue
while still getting a feasible solution. We can assume
w.l.o.g.\ that either $(\one^t D_\theta q)
\tilde{t}_\theta (q) > (\one^t D_\theta q) \tilde{t}_{\theta'} (q)$ or
$(x_\theta, \tilde{t}_\theta) = (x_{\theta'},
\tilde{t}_{\theta'})$.  Otherwise, if $(\one^t D_\theta q)
\tilde{t}_\theta (q) = (\one^t D_\theta q) \tilde{t}_{\theta'} (q)$, one can
simply replace $(x_\theta, \tilde{t}_\theta)$ with $(x_{\theta'},
\tilde{t}_{\theta'})$.

Now, consider the menu where the contract for a buyer of type $\theta$ is
$(x'_\theta, \tilde{t}'_\theta) = (x_\theta, (1-\epsilon) t_\theta)$. We claim
that for sufficiently small $\epsilon$, this is feasible and all constraints are
slack.

Any constraint that held strictly is still strict for sufficiently small
$\epsilon$. Any tight IR constraint becomes clearly strict for any $\epsilon >
0$; in addition, no IR constraint becomes violated.
Now for the IC constraints that were previously
tight and  $(x_\theta, \tilde{t}_\theta) \neq
(x_{\theta'}, \tilde{t}_{\theta'})$, we know that $(\one^t D_\theta q)
\tilde{t}_\theta (q) > (\one^t D_\theta q) \tilde{t}_{\theta'} (q)$, hence
the inequality becomes strict.
\end{proof}

\begin{lemma}
Given an optimal solution $(x_\theta, \tilde{t}_\theta)$ to $\mathbf{LP_2}$
where all the IC and IR constraints are slack except
possibly $IC_{\theta,\theta'}$ when $(x_\theta,
\tilde{t}_\theta) = (x_{\theta'}, \tilde{t}_{\theta'})$,
there is another solution $(x_\theta,
\tilde{t}'_\theta)$ with the same objective such that $x_\theta(q) > 0$ for all
$q \in \hat{Q}$.
\end{lemma}

\begin{proof}
Recall that $p \in \R^\Omega_+$ is the prior on $\omega$, i.e., $p(\omega) =
\mu(\omega)$. We can assume wlog that $p(\omega) > 0, \forall \omega\in \Omega$.
(Otherwise, we can simply restrict our attention to a subset of $\Omega$).
Therefore, $p$ is in the relative interior of $\Delta(\Omega)$. For all $q \in
\Delta(\Omega)$ there is $r(q) \in \Delta(\Omega)$ and $\gamma(q) \in [0,1]$ such
that:
$$p = \gamma(q) \cdot q + (1-\gamma(q)) \cdot r(q)$$
Now, consider $\tilde{t}'_\theta = \tilde{t}_\theta$ for all $\theta$ and
$x'_\theta(q) = (1-\delta) x_\theta(q)$. Now, for all $q \in \hat{Q}$ increase
$x'_\theta(q)$ by $\frac{1}{\abs{\hat{Q}}} \gamma(q) \delta$ and
$x'_\theta(r(q))$ by $\frac{1}{\abs{\hat{Q}}} (1-\gamma(q)) \delta$. If $\delta$
is sufficiently small, the IC and IR constraints are not violated. For the IC
constraints that held with equality and $(x_\theta,
\tilde{t}_\theta) = (x_{\theta'}, \tilde{t}_{\theta'})$, they still hold since
$(x'_\theta,
\tilde{t}'_\theta) = (x'_{\theta'}, \tilde{t}'_{\theta'})$. It is trivial to
check that by construction, the feasibility constraints still hold.
\end{proof}

Also, we would like to observe that as observed in Theorem \ref{thm:quadratic
support} for the Pricing Mappings Mechanism, the optimal protocol in this case
can be represented by a tree of $O(n^2+mn)$ size. Given the structural
similarities between $\mathbf{LP_1}$ and $\mathbf{LP_2}$, the proof follows by
a simple adaptation of the proof of Theorem \ref{thm:quadratic support}.
Formally:

\begin{observation}
Let $n = \abs{\Theta}$ and $m = \abs{\Omega}$. 
The program $\mathbf{LP_2}$ has a solution where
$\Vert x_\theta \Vert_0 \leq O(m+n)$ for all $\theta$.
\end{observation}

\section{Proofs omitted from Section~\ref{sec:continuity}}
\label{sec:ap-continuity}

This appendix contains the full details of the material
sketched in Section~\ref{sec:continuity}.

\subsection{Semicontinuity Proofs}
\label{subsec:ap-continuity-semicont}

\begin{theorem}\label{thm:lower-semicontinuity}
 The revenue function  $R:\mathcal{C} \rightarrow \R$ is lower-semicontinuos.
\end{theorem}

\begin{proof}
 Consider a sequence of contexts $(u_i, \mu_i)$ such that $(u_i, \mu_i)
\rightarrow (u, \mu)$, then consider a revenue optimal mechanism for the context
$(u, \mu)$. Invoking Lemma \ref{lemma:strict-preferences} we know that for every
$\epsilon > 0$, there is a mechanism for context $(u,\mu)$ generating revenue
at least $(1-\epsilon) \cdot R(u,\mu)$ such that all IR constraints are slack
and all IC constraints for non-identical-contracts are slack. Given that the
inequalities are strict and that the constraints are continuous in the context,
then for sufficiently large $i$, the mechanism is feasible for $(u_i,\mu_i)$,
so: $\liminf R(u_i, \mu_i) \geq R(u,\mu) - \epsilon$ for all $\epsilon > 0$.
\end{proof}

We note that a proof in the same spirit for the Cremer and McLean setting can
be found in Robert \cite{robert91}. Notice that the heart of the proof is the
application of Lemma \ref{lemma:strict-preferences}. In this specific
ingredient, our proof is different then the one of Robert.

The proof of Theorem \ref{thm:lower-semicontinuity} can be interpreted as a
robustness result: if the seller believes that the context $(u, \mu)$ and
designs an protocol extracting $R(u, \mu)$, but the real context is actually
slightly different $(u', \mu')$, the seller can design a protocol extracting
revenue $(1-\epsilon) \cdot R(u, \mu)$ from any context $(u', \mu')$ that is
sufficiently close to $(u, \mu)$.

This does not mean, however, that the seller is close to the
optimal revenue, since it might be the case that the a context $(u', \mu')$
allows the seller to extract a lot more revenue then $(u, \mu)$ even though the
two contexts are close.

\xhdr{Two Continuous Mechanisms}
In what follows we present two mechanisms that fix the two highlighted problems
of
the \emph{Pricing Outcomes Mechanism}: they are continuous and they do not
involve transfers from the seller to the buyers. Then we compare their
revenue-extration power with the optimal mechanism. For independent
$(\omega,\theta)$, they extract optimal revenue, but there are contexts with
correlated $(\omega,\theta)$ where
they extract only a $\Omega(\abs{\Theta}^{-1})$-fraction of the optimal
revenue.\\

\noindent \textbf{1. Pricing Outcomes with No Positive Transfers:} The auctioneer
offers to the buyers a menu of contracts $\{(Y_\theta, t_\theta)\}_{\theta}$,
where $Y_\theta$ is a random variable taking values in $\sspace_\theta$ and
$t_\theta: \sspace \rightarrow \R_+$. The only change with respect to the
optimal mechanism is that we require $t_\theta(s) \geq 0, \forall s \in
\sspace_\theta$. To find the optimal mechanism of this type, one can simply
solve $\mathbf{LP_2}$ with the additional
constraints that $\tilde{t}_\theta(q) \geq 0, \forall \theta,q$.\\

\noindent \textbf{2. Pricing Mappings:} The auctioneer
offers to the buyers a menu of contracts $\{(Y_\theta, t_\theta)\}_{\theta}$
where  $Y_\theta$ is a random variable taking values in $\sspace_\theta$ and
$t_\theta \in \R_+$ is a fixed value that the buyer needs to pay in order to
get the signal $Y_\theta$. The optimal mechanism of this type can be found using
a simplified version of $\mathbf{LP_2}$ :
simply substitute all the occurrences of $\sum_q (\one^t D_\theta q)
\tilde{t}_{\theta'} (q)$ by $\mu(\theta) \cdot t_{\theta'}$.\\

We presented four types of mechanisms which are increasing in complexity and
revenue-extracting power. The following diagram expresses which types are a
special case of which:

{ \small
$$\left\{ \begin{array}{c} \text{Sealed} \\ \text{Envelope} \end{array} \right\}
\subseteq
\left\{ \begin{array}{c} \text{Pricing} \\ \text{Mappings} \end{array} \right\}
\subseteq
\left\{ \begin{array}{c} \text{Pricing Outcomes} \\ t\geq 0
\end{array} \right\}
\subseteq
\left\{ \begin{array}{c} \text{Pricing} \\ \text{Outcomes} \end{array}
\right\}$$
}

For a given context $(u,\mu) \in \mathcal{C}$, we define $R_e(u,\mu)$,
$R_c(u,\mu)$ and $R_p(u,\mu)$ as the optimal revenue that can be obtained using
a Sealed Envelope Mechanism, Pricing Mappings Mechanism and Pricing Outcomes
With No Positive Transfers Mechanism, respectively. We know that:

$$R_e(u,\mu) \leq R_c(u,\mu) \leq R_p(u,\mu) \leq R(u,\mu)$$

And by Theorem
\ref{thm:rev-principle-independent} we know that for independent $\mu$, i.e,
$\mu(\omega,\theta) = \mu(\omega) \cdot \mu(\theta)$, we have:

$$R_e(u,\mu) \leq R_c(u,\mu) = R_p(u,\mu) = R(u,\mu)$$

First we will show that $R_p$, $R_c$ and $R_e$ are continuous in $\mathcal{C}$
and then we analyze the gap in their revenue extraction power. By a similar
argument then the one in Theorem
\ref{thm:lower-semicontinuity}, it is simple to see that they are
lower-semicontinuous. In order to prove continuity, we only need to show
upper-semicontinuity. We say that a function $f:\R^k \rightarrow \R$ is
{\em upper-semicontinuous} if for
for all converging sequences $x_i \rightarrow x$, $f(x) \geq \limsup f(x_i)$. A
function is continuous if it is both lower-semicontinuous and
upper-semicontinuous.

\begin{lemma}\label{lemma:continuity_Rc}
 The revenue function $R_c:\mathcal{C} \rightarrow \R$ is
upper-semicontinuous.
\end{lemma}

\begin{proof}
Consider a context $(u, \mu) \in \mathcal{C}$ and $(u_i, \mu_i) \rightarrow
(u,\mu)$. Passing to a subsequence if necessary we can assume that $R_c(u_i,
\mu_i)$ converges to the $\limsup$ of the original sequence. Now, we need to
prove that $R_c(u,\mu) \geq \lim R_c(u_i, \mu_i)$. We do so in two steps: (1)
we show that the optimal Pricing Mappings Mechanism can be represented by an
element in a compact set; (2) we use compactness arguments to take the limit of
such mechanism obtaining a mechanism for context $(u,\mu)$.

{\em Compact representation:} The revenue $R_c(u_i,\mu_i)$ is the solution of a
variant the $\mathbf{LP_2}$ described above. As the original $\mathbf{LP_2}$,
it has $O(\abs{\Theta}^2)$ constraints and therefore has an optimal solution
that has support $O(\abs{\Theta}^2)$. Now, the optimal mechanism for context
$(u_i,\mu_i)$ can be represented by a price vector $t^i \in \R^\Omega_+$ and
random variables $Y_\theta^i$ with  $O(\abs{\Theta}^2)$ support. This
means that each of those random variables can be represented as $\abs{\Omega}$
distributions $(\psi^\omega_\theta)^i \in \Delta([O(\abs{\Theta}^2)])$. So, the
set of variables $\{Y_\theta^i\}_{\theta \in \Theta}$ admits a representation
as an element of a compact set. Now, we need to argue that the set of vectors
$t^i$ are bounded and therefore can be also represented as elements of a
compact set.

As for the price vector $t^i$, notice that $t^i \leq \buyersurplus^i$ for any
individually rational mechanism, where $\buyersurplus^i$ is the surplus vector
for
context $(u_i, \mu_i)$. If $\buyersurplus$ is the surplus vector for context
$(u,\mu)$ then for $i \geq i_0$, then $\Vert \buyersurplus^i - \buyersurplus
\Vert_\infty \leq 1$. For the $i \geq i_0$, therefore $t^i$ is in the compact
$\prod_\theta [0,\buyersurplus(\theta)+1]$.

{\em Limit mechanism:} Now, for each context $(u_i, \mu_i)$, the optimal
mechanism can be represented a point $\mathcal{M}_i \in K$ where $K$ is a
compact set. Therefore, passing to a subsequence if necessary, we can assume
that $\mathcal{M}_i \rightarrow \mathcal{M}$. Now, since IR and IC properties
are preserved in the limit (since they are continuous properties),
$\mathcal{M}$ is a valid Pricing Mappings Mechanism for context $(u,\mu)$ and
its revenue is at least the limit of the revenue extracted by $\mathcal{M_i}$
on $(u_i,\mu_i)$. Therefore, $R_c(u,\mu) \geq \limsup R_c(u_i,\mu_i)$.
\end{proof}

An even simpler compactness argument can be formulated for $R_e$. For $R_p$ we
need to add a mild non-degeneracy assumptions: we say that a context $(u,\mu)$
is non-degenerated if $\mu(\omega, \theta) > 0, \forall \omega,\theta$.

\begin{lemma}\label{lemma:continuity_Rp}
The revenue function $R_p : \mathcal{C} \rightarrow \R$ is
upper-semicontinuous in the neighborhood of all non-degenerated context.
\end{lemma}

\begin{proof}
As in the proof of Lemma \ref{lemma:continuity_Rc}, the main goal is to show
that the set of optimal mechanisms for contexts $(u_i, \mu_i)$ are members of a
compact set and then proceed as in the second part of the proof of that lemma.
For the first part, we proceed for the $Y_\theta^i$ as in Lemma
\ref{lemma:continuity_Rc}. But we need a different argument to show that the
transfers in the optimal mechanism are bounded.

 Given a non-degenerated context, let $\delta = \frac{1}{2}
\min_{\omega,\theta} \mu(\theta \vert \omega)$, then if $(u_i, \mu_i)
\rightarrow (u,\mu)$, then for $i \geq i_0$, then $\mu_i(\theta \vert \omega)
\geq \delta$ for all $\omega,\theta$. Increasing $i_0$ if necessary, we
can also assume that $\Vert \buyersurplus^i - \buyersurplus \Vert_\infty \leq 1$
for $i
\geq i_0$. So, if
$t^i_\theta(q)$ is the payment of buyer $\theta$ for posterior $q$ in the
optimal mechanism for context $(u_i, \mu_i)$, then $\tilde{t}^i_\theta(q) =
x^i_\theta(q) \cdot t^i_\theta(q)$ is in the compact set $\prod_i[0, \frac{1}{
\delta} (\buyersurplus(\theta) + 1)]$ for all individually rational mechanism,
since:

$$\sum_q v_\theta(D_\theta q) x_\theta(q) - (\one^t D_\theta q)
\tilde{t}_\theta (q) \geq v_\theta(D_\theta p)$$
implies that:
$$\begin{aligned}
\buyersurplus(\theta) +1 & \geq  \buyersurplus^i(\theta)  \geq
\frac{1}{\mu^i(\theta)}
\left[ \sum_q v^i_\theta(D_\theta q) x^i_\theta(q) - v^i_\theta(D_\theta p)
\right] \\ & \geq \sum_q (\one^t D_\theta q) \cdot
\tilde{t}^i_\theta (q) \geq (\one^t D_\theta q) \cdot
\tilde{t}^i_\theta (q) \geq \delta \cdot \tilde{t}^i_\theta (q)
\end{aligned}$$

Notice that inequalities crucially depend on $\tilde{t}^i_\theta(q)$ being
non-negative.
\end{proof}

The following theorem summarizes the previous discussion:

\begin{theorem}
The revenue functions $R_c, R_e : \mathcal{C} \rightarrow \R$ are
continuous and the revenue function $R_p:\mathcal{C} \rightarrow \R$ is
continuous around any non-degenerated context.
\end{theorem}

\subsection{Approximating Revenue}
\label{sec:ap-continuity-approx}

We know from Example \ref{example:contracts-vs-envelope} that the Sealed
Envelope Mechanism performs very poorly in terms of revenue when compared to
the optimal mechanism, even if the $\omega,\theta$-signals are independent. A
natural question is how do the other mechanisms perform in comparison to the
optimal mechanism. The Pricing Mappings Mechanism is simple and
allows for a much more natural implementation then the Pricing Outcomes.
Unfortunately, as we show in this section, both the Pricing Mappings Mechanism
and the Pricing Outcomes with no positive transfers Mechanism can perform very
poorly compared to the Pricing Outcomes Mechanism for certain contexts:

\begin{lemma} The revenue extracted by a Pricing Mappings Mechanism and by a
Pricing Outcomes with no positive transfers Mechanism is at least a
$1/\abs{\Theta}$-fraction of the optimal mechanism:
$$R_c(u,\mu) \geq \frac{1}{\abs{\Theta}} R(u,\mu), \quad R_p(u,\mu) \geq
\frac{1}{\abs{\Theta}} R(u,\mu)$$ and those bounds are tight up to constant
factors.
\end{lemma}

The bound is trivial and it follows directly from the fact that the optimal
revenue is bounded by the surplus, and extracting full surplus from a single
type is trivial. Now, we show that the bound is tight using the continuity of
$R_c(\cdot)$ and $R_p(\cdot)$.

\begin{example}
In this example we show that there exist a context $(u,\mu)$ such that
$R_c(u,\mu) \leq O(1/\abs{\Theta}) R(u,\mu)$ and $R_p(u,\mu) \leq
O(1/\abs{\Theta}) R(u,\mu)$.

For $\Omega = \Theta = [n]$, we present a non-degenerated context $(u,\mu)$
where
$(\omega,\theta)$
are independent and the optimal revenue is a $1/n$-fraction of the
full surplus $\sum_\theta \mu(\theta) \buyersurplus(\theta)$. For this context:
$$R(u,\mu) = R_p(u,\mu) = R_c(u,\mu) = \frac{1}{n} \sum_\theta \mu(\theta)
\buyersurplus(\theta)$$
Now, if one sees $\mu$ as a $n \times n$-matrix and $\eta$ is any
$n \times n$-matrix of full rank with $\sum_{\omega,\theta} \eta(\omega,\theta)
= 0$, then
for any small $t>0$, $\mu+t \eta$ represents a probability distribution
over $\Theta \times \Omega$ close to $\mu$ but with $\text{rank}(\mu+t \eta) =
n$,
therefore, the optimal mechanism extracts full revenue (Theorem
\ref{thm:full-surplus}). In other words:
$$R(u,\mu+t\eta) = \sum_\theta \mu(\theta)
\buyersurplus_t(\theta)$$
where $\buyersurplus_t$ is the surplus for context $(u, \mu+t \eta)$. By
continuity:
$$\lim_{t \rightarrow 0+} R_c(u,\mu+t\eta) = R_c(u,\mu) = \frac{1}{n}
\sum_\theta \mu(\theta)
\buyersurplus(\theta)$$
while
$$\lim_{t \rightarrow 0+} R(u,\mu+t\eta) = \sum_\theta \mu(\theta)
\buyersurplus(\theta)$$
showing that the bound it tight. So, all we need to so is to exhibit a context
where the optimal revenue is a $\beta/n$ fraction of the full surplus for some
constant $\beta$.

Let $\Omega = \Theta = A = [n]$, $u(\theta,\omega,a) = v_\theta \cdot \one \{a
= \omega\}$, $v_\theta = 2^\theta$, $\mu(\omega) = \frac{1}{n}$ and $\mu(\theta)
= \frac{2^{-\theta}}{1-2^{-n}}$. The surplus is given by $\buyersurplus(\theta)
=
2^\theta (1-\frac{1}{n})$, so $\sum_\theta \mu(\theta) \buyersurplus(\theta) =
\frac{n-1}{1-2^{-n}}$. The revenue optimal mechanism for this context is a
Pricing Mappings Mechanism.

Using the construction in the Interesting Posteriors Lemma (Lemma
\ref{lemma:interesting-posteriors}) we can search for a solution where the only
posteriors in the support are of the form $q_X(\omega) = \frac{1}{\abs{X}}$ for
$\omega \in X$ and $q_X(\omega) = 0$ otherwise, where $X$ is any subset of
$[n]$. Therefore we can solve the primal LP only considering the
$x_\theta(q_X)$ variables, i.e., we need to solve:
$$
\begin{aligned}
& \max \sum_\theta \frac{2^{-\theta}}{1-2^{-n}} t_\theta \text{ s.t. }\\
& \quad \begin{aligned}
& -t_\theta + \sum_{X \subseteq [n]} x_\theta(q_X) \frac{v_\theta}{\abs{X}}
\geq \frac{v_\theta}{n}, \\
& -t_\theta + \sum_{X \subseteq [n]} x_\theta(q_X) \frac{v_\theta}{\abs{X}}
\geq -t_{\theta'} + \sum_{X \subseteq [n]} x_{\theta'}(q_X)
\frac{v_\theta}{\abs{X}},
 \\
& \sum_\theta x_\theta(q_X) = 1,  x_\theta(q_X) \geq 0 \\
\end{aligned} \end{aligned}
$$

Notice one can define $\hat{x}_\theta = \sum_{X \subseteq [n]}
\frac{1}{\abs{X}} x_\theta(q_X)$ and re-write the LP as:

$$
\begin{aligned}
& \max \sum_\theta \frac{2^{-\theta}}{1-2^{-n}} t_\theta \text{ s.t. }\\
& \quad \begin{aligned}
& -t_\theta + \hat{x}_\theta v_\theta
\geq v_\theta /n, \\
& -t_\theta + \hat{x}_\theta v_\theta
\geq -t_{\theta'} +\hat{x}_{\theta'} v_\theta,
 \\
& 1/n \leq \hat{x}_\theta \leq 1 \\
\end{aligned} \end{aligned}
$$

Now it suffices to show that the solution of the above LP is bounded by a
constant as $n$ grows. This can be seen very easily by showing a feasible dual.
We have that:

$$\left[ -t_1 + \hat{x}_1 v_1
- \frac{v_1}{n} \right] + \sum_{\theta=1}^{n-1} 2^{-\theta} \left[
-t_{\theta+1} + \hat{x}_{\theta+1} v_{\theta+1}
 + t_{\theta} -\hat{x}_{\theta} v_{\theta+1} \right] \geq 0$$
substituting $v_\theta = 2^\theta$ and simplifying the expression, we get:

$$\sum 2^{-\theta} t_\theta \leq 2 \hat{x}_n - \frac{t_n}{2^n} - \frac{2}{n}
\leq 2 $$

\end{example}

One might also wonder what is the relation between $R_c$ and $R_p$. Clearly
$R_c(u,\mu) \geq \frac{1}{\abs{\Theta}} R_p(u,\mu)$. The following simple
example shows that this bound is tight.

\begin{example}
\label{ex:rcrp}
 In this example we show that there exists a context $(u,\mu)$
such that $R_c(u,\mu) \leq O(1/\abs{\Theta}) R_p(u,\mu)$.

We consider in this example
the case where the seller learns
exactly the type of the buyer. We encode this information in his signal
$\omega$. It will have two components: $\omega_1$ which is the information the
buyer is interested and $\omega_2$ which tells the seller the type $\theta$ of
the buyer. A Pricing Outcomes Mechanism will be able to explore this
information in order to extract full surplus from the buyer, but a pricing
contracts cannot.

Let $\Omega_0 = \{0,1\}$ and $\Theta_0 = [n]$ and consider a probability
distribution $\mu_0$ on $\Omega_0 \times \Theta_0$ where
$\mu_0(\omega_0,\theta_0) = \frac{1}{2^{1-\theta_0}} \cdot \frac{1}{1-2^{-n}}$.
Now, define utilities of the buyer as functions of $u_0(\omega_0, \theta_0,a)$.
As usual, we define as a piecewise-linear function $v_{0,\theta_0}:[0,1]
\rightarrow \R$. Consider the function interpolating
$(0,2^{\theta_0}),(\frac{1}{2},0),(1,2^{\theta_0})$. Now, let's define $\Theta =
\Theta_0$ and $\Omega = \Omega_0 \times \Theta_0$, and if $\omega = (\omega_1,
\omega_2)$, then we define a probability distribution $\mu$ on $\Omega \times
\Theta$ as $\mu(\omega,\theta) = \mu_0(\omega_1,\theta)$ and $u(\omega,\theta,a)
= u_0(\omega_1, \theta, a)$.

By using a Pricing Mappings Mechanism, the seller is not able to use the
information in $\omega_2$, and therefore he can only extract $O(1)$ revenue. If
the seller is able to price signals, he can design the contract $\theta$ such
that the price of $(0,\theta)$ and $(1,\theta)$ is $2^\theta-\epsilon$ and the
price of the signals $(\omega_0, \theta')$ is $\infty$ for $\theta'\neq
\theta$. Therefore, buyers will be forced to buy the contract for their correct
type.

The implementation above casts the mechanism in the general ``Pricing Outcomes''
form given by Theorem \ref{thm:rev-principle-correlated}, but a more natural
implementation would
be for the seller to learn the $\theta = \omega_2$ and then announce price
$2^\theta - \epsilon$ for the full information signal.
\end{example}

\section{Proofs omitted from Section~\ref{sec:uncommitted}}
\label{sec:ap-uncommitted}

\begin{claim}
\label{lem:uncommitted_separation}
In the interactive protocol presented in Example~\ref{example:uncommitted_separation}
the optimal strategy for an uncommitted buyer is to
play left (to node $t_1$) whenever $\theta=0$ and
play right (to node $s_2$)
whenever $\theta=1$, and then follow the protocol (make the transfers when asked) without defecting.
\end{claim}
\begin{proof}
It is simple to see that a buyer of type $\theta=1$ will not play left
at the root, since his surplus for the information is only $\buyersurplus(0) =
0.4 <
0.533$.
A buyer of type $\theta=0$ playing left and not defecting gets the full
information and pays $0.533$, so his utility is $1-0.533=0.467$.
information and pays $0.533$, so his utility is $1-0.533=0.467$.
To complete the proof we still need to show that a buyer of type $\theta=0$
would not prefer to
play right 
and that the buyer of type $\theta=1$ will follow the entire protocol.

In the node $s_2$, upon receiving $\ell_5$ both buyer types update their
beliefs to $\omega=1$. Now, upon receiving $t_2$, they update their beliefs
about $\omega$ to:
$\prob(\omega \vert \theta, t_2) =  \mu(\omega \vert \theta)~\cdot~\prob(t_2
\vert \omega) / \sum_{\omega'} \left(\mu(\omega' \vert \theta) \cdot
\prob(t_2 \vert \omega')\right) $.
Doing the calculations, we get that $\prob(\omega=1\vert\theta=0,t_2) =
0.1$ and $\prob(\omega=1 \vert \theta=1, t_2) = 0.2$. By the shape of the
$v_\theta(\cdot)$ function, no player has increase in utility by this signal.
Hence, a buyer of type $\theta=0$ would not play right and then defect in $t_2$.
At this point (node $t_2$) the surplus for knowing the value of $\omega$
exactly is $0.9$ for $\theta=0$ and $0.8$ for $\theta=1$. So, it is clear that
type $\theta=1$ would pay the amount of $0.8$ to get the signal revealed.

Now, we are left with the question of whether the player of type $\theta=0$
would not prefer to play right. 
In both cases (left and right), he gets
the full information about the value of $\omega$. Now, we need to compare which
one has the smallest expected price in his perspective.
The price of the left path is $0.533$ and the left path has price $0.8 \cdot
\prob(t_2 \vert
\theta=0) = \frac{2}{3} \cdot 0.8 = 0.5333..$.
\end{proof}

\subsection{Proof of Theorem~\ref{thm:gap-committed-uncommitted}}
To prove Theorem~\ref{thm:gap-committed-uncommitted} we consider the setting of 
Example \ref{example:uncommitted_separation}.

We first observe that it is possible to extract revenue $0.5$ from committed
buyers using a Pricing Outcomes with No Positive Transfers Mechanism. The
interim
belief of buyer of type $0$ signal is $\mu(\omega=1\vert \theta=0) = 0.4$, so
his surplus for the information is $\buyersurplus(0) = 0.6$, while for buyer of
type
$1$ his interim belief is $0.6$ and his surplus is  $\buyersurplus(1) = 0.4$.
By Theorem \ref{thm:full-surplus} we can extract the full surplus
$\frac{1}{2}\buyersurplus(0)+\frac{1}{2}\buyersurplus(1)=0.5$ from {\em committed}
buyers. This can be done by a Pricing Outcomes with No Positive Transfers
Mechanism
in which the seller reveals the value of $\omega$ and charges $t(\omega)$, where $t(0) = 1$ and $t(1) = 0$.

The next lemma completes the proof of the theorem.

\begin{lemma}\label{lem:gap-committed-uncommitted}
For the context $(u,\mu)$ with correlated $\omega$ and $\theta$ that is defined in Example
\ref{example:uncommitted_separation} it hold that there exists some $\delta>0$ such that
the revenue of any protocol for uncommitted buyers that has no positive transfers is at most $0.5-\delta$.
\end{lemma}
\begin{proof}
Fix $\delta = 10^{-10} \ll 10^{-2} = \epsilon$.
Let's assume that there is a protocol with no positive transfers that is able to extract revenue
$0.5-\delta$ from uncommitted buyers. We can assume by Theorem
\ref{thm:uncommitted_partial_revelation} that the protocol has the following
format: the root is a buyer node with two children and all the other nodes of
the tree are either seller or transfer nodes. We can also assume that the
optimal strategy for the $\theta=0$ buyer is to take the left path and follow
the protocol until then end and that the optimal strategy for the $\theta=1$
buyer is to take the right path and follow the protocol until the end.

Now, since the surplus of the types are $0.6$ and $0.4$ respectively, the
protocol should extract at least $0.6-2\delta$ from buyer of type $0$ and at
least $0.4-2\delta$ from buyer of type $1$. We want to claim that in such a
setting, a buyer of type $\theta = 0$ can benefit from taking the right
path (i.e., behaving like a buyer of type $\theta=1$) paying strictly less
then $0.6-2\delta$. This generates a contradiction.

Now, we start analyzing the change in utility for $\theta=0$ if he does as
$\theta=1$, i.e. takes the right path.  Let $L$ be the set of leaves that are in
the right side of the root. And let
$\psi^\omega \in \Delta(L)$ for be the distributions induced on
the leaves for each $\omega \in \{0,1\}$. Now, if the output of the protocol is
leaf $\ell$ the buyer of type $\theta$ updates his belief to:
$$\prob(\omega \vert \ell, \theta) = \frac{\mu(\omega \vert \theta)
\psi^\omega(\ell)}{\sum_{\omega'} \mu(\omega' \vert \theta)
\psi^{\omega'}(\ell)}$$
In vector notation, the posterior of a buyer of type $\theta$ upon
observing a leaf $\ell$ of the right path is:
$$q^{\theta,\ell} = \frac{A^\ell q^\theta}{\one^t A^\ell q^\theta}, \text{
where } A^\ell = \begin{bmatrix} \psi^0(\ell) & \\ & \psi^1(\ell) \end{bmatrix}
\text{ and } q^\theta = \begin{bmatrix}  \mu(\omega=0 \vert\theta) \\
\mu(\omega=1 \vert\theta) \end{bmatrix}$$
He obtains this posterior with probability $\sum_{\omega'} \mu(\omega' \vert
\theta) \psi^{\omega'}(\ell) = \one^t A^\ell q^\theta$. Therefore, his total
value (not counting payments) for taking the right path is: $\sum_\ell (\one^t
A^\ell q^\theta) \cdot v ( \frac{A^\ell q^\theta}{\one^t A^\ell q^\theta} ) $.
The previous expression has a slightly notation abuse: we defined $v:[0,1]
\rightarrow \R_+$. We overload the function $v$ with $v:\Delta(\Omega)
\rightarrow \R_+$ such that for $[\begin{smallmatrix}
q(0)\\q(1)\end{smallmatrix}]$ we define $v(q) = v(q(1))$.

\begin{fact}\label{fact:high_value_deviation} If $\sum_\ell (\one^t
A^\ell q^1) \cdot v ( q^{1,\ell} )
\geq 1-2\delta$, then it must be the case that $\sum_\ell (\one^t
A^\ell q^0) \cdot v ( q^{0,\ell} )
\geq 1-36 \cdot \delta$. In other words, if a buyer of type $\theta=1$ can
extract value very close to $1$ from learning the leaves reached by the
protocol, then a buyer of type $\theta=0$ should also be able to extract values
very close to $1$ from it.
\end{fact}

\begin{proof}
First let $L_1 = \{\ell; \sqrt{2\delta} < q^{1,\ell}(1) <
1-\sqrt{2\delta}\}$. First, we claim that $\sum_{\ell \in L_1} \one^t A^\ell
q^1 < \sqrt{2\delta}$. If not, then:
$$\begin{aligned}1-2\delta & < \sum_\ell (\one^t
A^\ell q^1) \cdot v ( q^{1,\ell} ) \leq \sum_{\ell \in L_1} (\one^t
A^\ell q^1) \cdot v ( q^{1,\ell} ) + \sum_{\ell \notin L_1} (\one^t
A^\ell q^1) \cdot v ( q^{1,\ell} ) \leq \\
& \leq (1-\sqrt{2\delta}) \sum_{\ell \in L_1} (\one^t
A^\ell q^1) + 1 \cdot \sum_{\ell \notin L_1} (\one^t
A^\ell q^1) = 1 - \sqrt{2\delta} \sum_{\ell \in L_1} (\one^t
A^\ell q^1) \leq \\
& \leq 1 -  \sqrt{2\delta}  \cdot  \sqrt{2\delta}  = 1-2\delta
\end{aligned}$$
using the fact that $v(q) \leq \max\{q,1-q\}$ and the fact that $\sum_\ell
A^\ell = [\begin{smallmatrix} 1 & \\ & 1 \end{smallmatrix}]$ by the definition
of $A^\ell$. Now, in order to get an upper bound for $\sum_\ell (\one^t
A^\ell q^0) \cdot v ( q^{0,\ell} )$  we focus on $\ell \notin
L_1$.  For those vectors, we still need to bound  $v ( q^{0,\ell} )$. We
consider two cases.
\begin{itemize}
 \item \textbf{case 1:} $q^{1,\ell}(1) \leq \delta' := \sqrt{2\delta}$.
$$\delta' \geq q^{1,\ell}(1) = \frac{A^\ell_{11} q^1(1)}{A^\ell_{00}
q^1(0)+A^\ell_{11} q^1(1)} = \frac{A^\ell_{11} 0.6}{A^\ell_{00} 0.4+A^\ell_{11}
0.6} \Rightarrow \frac{A^\ell_{00}}{A^\ell_{11}} \geq \frac{3
(1-\delta')}{2\delta'} \geq \frac{1}{2\delta'} $$
Now, calculating $q^{0,\ell}(1)$, we get:
$$ q^{0,\ell}(1) = \frac{A^\ell_{11} q^0(1)}{A^\ell_{00}
q^0(0)+A^\ell_{11} q^0(1)} = \frac{A^\ell_{11} 0.4}{A^\ell_{00} 0.6+A^\ell_{11}
0.4} \leq  \frac{ 0.4}{(1/2\delta') 0.6+ 0.4} \leq 2 \delta' $$
Therefore:
$v(q^{0,\ell}) \geq 1-9(2 \delta') = 1-18 \cdot \delta'$
 \item \textbf{case 2:} $q^{1,\ell}(1) \geq 1-\delta' = 1-\sqrt{2\delta}$.
$$\delta' \geq q^{1,\ell}(0) = \frac{A^\ell_{00} q^1(0)}{A^\ell_{00}
q^1(0)+A^\ell_{11} q^1(1)} = \frac{A^\ell_{00} 0.4}{A^\ell_{00} 0.4+A^\ell_{11}
0.6} \Rightarrow \frac{A^\ell_{11}}{A^\ell_{00}} \geq \frac{2
(1-\delta')}{3\delta'} \geq \frac{1}{3\delta'} $$
Now, calculating $q^{0,\ell}(0)$, we get:
$$ q^{0,\ell}(0) = \frac{A^\ell_{00} q^0(0)}{A^\ell_{00}
q^0(0)+A^\ell_{11} q^0(1)} = \frac{A^\ell_{00} 0.6}{A^\ell_{00} 0.6+A^\ell_{11}
0.4} \leq  \frac{ 0.6}{ 0.6+ 0.4(1/3\delta')} \leq 5 \delta' $$
Therefore:
$v(q^{0,\ell}) \geq 1-(5 \delta')$
\end{itemize}
The important fact is that
$v(q^{0,\ell}) \geq 1-18 \sqrt{2\delta}$ for each $\ell \notin L_1$. Now,
we are ready to bound the value that type $\theta=0$ gets from the right path:
$$\sum_{\ell \notin L_1} (\one^t
A^\ell q^0) \cdot v ( q^{0,\ell} ) \geq \sum_{\ell\in L_1} (\one^t
A^\ell q^0) \cdot v ( q^{0,\ell} ) \geq (1-\sqrt{2\delta})\cdot(1-18
\sqrt{2\delta}) \geq 1-36 \cdot \delta$$
\end{proof}

Now, for the next step, we bound the total expected payment of buyer $\theta=0$
if he decides to deviate to the right path and follow the protocol until the
end. In order to do so, let $N$ the set of transfer nodes in the left path and
for each $i \in N$, let $t_i$ be the amount associated with it. Also, let
$\prob(i\vert \omega)$ be the probability that a buyer following the right
branch of the tree reaches $i$ condioned on $\omega$. Define also the matrix:
$$A^i = \begin{bmatrix} \prob(i\vert \omega=0) & 0\\ 0& \prob(i\vert \omega=1)
\end{bmatrix}$$

Now, remember we had fixed in the beginning a constant $\epsilon \gg \delta$ but
still $\epsilon \ll 1$. Now the expected payment of player of type $\theta$ in
the right branch of the protocol is given by:
$$\sum_\omega \mu(\omega \vert \theta) \sum_i \prob(i \vert \omega) t_i =
\sum_i (\one^t A^i q^\theta) \cdot t_i$$
We know that player of type $\theta=1$ pays at least $0.4-2\delta$, so $\sum_i
(\one^t A^i q^1) \cdot t_i \geq 0.4-2\delta$. Now, we show that the payment of
buyer $\theta=0$ if he switches to the right branch of the protocol (i.e. moves
as a buyer of type $\theta=1$) is $1-O(\epsilon)$, giving him utility
$O(\epsilon) \gg 2\delta$ for doing so. We formalize this below.

\begin{fact}\label{small:small_payment_deviation}
 If the payment of player of type $\theta=1$ is $\sum_i (\one^t A^i q^1) \cdot
t_i \leq 0.4$, then the expected payment $\sum_i (\one^t A^i q^0) \cdot
t_i$ of the player of type $\theta=0$ after moving is at most $0.6
-\epsilon(0.4-\epsilon)$.
\end{fact}

\begin{proof}
We will decompose the sum $\sum_i (\one^t A^i q^0) \cdot t_i$ in two parts and
bound each of them separately. First, we divide the transfer nodes in two
classes:
$$N_1 = \{i \in N; \one^t A^i q^0 \leq (\tfrac{3}{2}-\epsilon) \one^t A^i
q^1\}$$
$$N_2 = \{i \in N; \one^t A^i q^0 > (\tfrac{3}{2}-\epsilon) \one^t A^i q^1\}$$

First note that for all $i$ (even $i \in N_2$), $\one^t A^i q^0 \leq
\tfrac{3}{2}\one^t A_i q^1$ simply because $q^0 \leq \tfrac{3}{2} q^1$
component-wise. Now, we can write:

\begin{equation}\label{eq:payment_decomposition}\sum_i (\one^t A^i q^0) \cdot
t_i \leq \sum_{i \in N_1} (\tfrac{3}{2}-\epsilon) \cdot \one^t A^i q^1 + \sum_{i
\in N_2}
\tfrac{3}{2} \cdot \one^t A^i q^1\end{equation}
If we show that most of $\sum_i (\one^t A^i q^1) \cdot t_i$ is concentrated on
$N_1$ we are done.

Now, consider some transfer node $i \in N_2$ and recall that $q^0 =
[\begin{smallmatrix} 0.6\\0.4 \end{smallmatrix}]$ and $q^1 =
[\begin{smallmatrix} 0.4\\0.6 \end{smallmatrix}]$. Therefore:
\begin{equation}\label{eq:diagonal_ratio}i \in N_2 \Leftrightarrow 0.6 A^i_{00}
+ 0.4 A^i_{11} \geq (\tfrac{3}{2}-\epsilon) (0.4 A^i_{00} + 0.6 A^i_{11})
\Leftrightarrow \frac{A^i_{11}}{A^i_{00}} \leq
\frac{4\epsilon}{5-6\epsilon}\end{equation}

Observe also that once in a node $i \in N_2$, a buyer of type $\theta=1$ is
willing to
follow the protocol until the end, the expected amount he is paying by
following this node, must be at least as big as his surplus for the full
information at that node. Formally:
\begin{equation}\label{eq:voluntary_participation}1 - v(\frac{A^i q^1}{\one^t
A^i q^1}) \geq t_i + \sum_{j \in N(i)} t_j \cdot \frac{\one^t A^j q^1}{\one^t
A^i q^1} \end{equation}
where $N(i)$ is the set of transfer nodes that have $i$ as an ancestor
(excluding $i$ itself). In order to justify the term $\frac{\one^t A^j
q^1}{\one^t A^i q^1}$, notice that if $i$ is an ancestor of $i$ then the
probability of reaching $j$ given $i$ is given by:
$$\prob(j \vert i,\theta) = \frac{\prob(i,j \vert \theta)}{\prob(i \vert
\theta)} = \frac{\prob(j \vert \theta )}{\prob(i \vert \theta)} = \frac{\one^t
A^j q^1}{\one^t A^i q^1} $$
 Now, given $i \in N_2$, since
the matrix $A^i$ is of the form $[\begin{smallmatrix} O(1) & \\ & O(\epsilon)
\end{smallmatrix}]$, the posterior probability of each player in that node is
$O(\epsilon)$, so $v(q) = w^t q$ where $w =
[\begin{smallmatrix} 1\\-8\end{smallmatrix}]$. Therefore, for $i \in
N_2$, we can re-write equation (\ref{eq:voluntary_participation}) as:
$$   t_i \cdot (\one^t A^i q^1) + \sum_{j \in N(i)}
t_j \cdot (\one^t A^j q^1) \leq \hat{w}^t A^i q^1$$ where $\hat{w} =
\one-w = [\begin{smallmatrix} 0\\9\end{smallmatrix}]$. Now, let $N_2^g$ be the
nodes in $N_2$ that have no ancestor in $N_2$.
Therefore:
$$\sum_{i \in N_2} t_i \cdot (\one^t A^i q^1) \leq \sum_{i \in N_2^g} \left[
t_i \cdot (\one^t A^i q^1) + \sum_{j \in N(i)} t_j \cdot (\one^t A^j q^1)
 \right] \leq \sum_{i \in N_2^g} \hat{w}^t A^i q^1 \leq 9 \cdot \sum_{i \in
N_2^g} A^i_{11} \cdot 0.6$$
Now, we know by equation (\ref{eq:diagonal_ratio}) that:
$$\sum_{i \in N_2^g} A^i_{11} \leq \sum_{i \in N_2^g}
\frac{4\epsilon}{5-6\epsilon} \cdot A^i_{00} \leq \epsilon \sum_{i \in N_2^g}
\prob(i \vert \omega=0) \leq \epsilon$$
where we are using the definition that $A^i_{00}$ is the probability that the
protocol reached node $i$ given $\omega=0$ and we are using the fact that
$\sum_{i \in N_2^g} \prob(i \vert \omega=0) \leq 1$ since the protocol reaching
$i \in N_2^g$ are disjoint events, since for all paths in the protocol-tree
there is at most one node in $N_2^g$. Therefore: $\sum_{i \in N_2} t_i \cdot
(\one^t A^i q^1) \leq 9 \epsilon$, which allows us to finish the bound in
equation~(\ref{eq:payment_decomposition}):
$$\sum_i (\one^t A^i q^0) \cdot
t_i \leq \sum_{i \in N_1 \cup N_2} (\tfrac{3}{2}-\epsilon) \cdot \one^t A^i q^1
+ \epsilon \sum_{i\in N_2} \cdot \one^t A^i q^1 \leq 0.4 (\tfrac{3}{2}-\epsilon)
- \epsilon^2 = 0.6 - \epsilon(0.4-\epsilon)$$
\end{proof}

Taking Fact \ref{fact:high_value_deviation} and Fact
\ref{small:small_payment_deviation} together, it is clear that there cannot be a
protocol extracting $0.5-\delta$ from uncommitted buyers without transfers from
the seller to the buyer. Summarizing the argument, if there were such protocol,
then the buyer of type $\theta=1$ would be faced with the option of getting
almost full information (i.e. getting a signal that would give him at least
$0.4-2\delta$ value) and would pay at most $0.4$. If a buyer of type $\theta=0$
pretends to be of type $\theta=1$, he guarantees by fact Fact
\ref{fact:high_value_deviation} a contract that has almost full information
(gives him value at least $1-36\delta$) and charges him at most $0.6 -
0.4\epsilon + \epsilon^2$ (Fact \ref{small:small_payment_deviation}). This
contradicts the fact that we can extract at least $0.6-2\delta$ revenue from
type $\theta=0$.
\end{proof}

%
%

\end{document}